\documentclass[11pt]{article}
\usepackage[margin=1.5in]{geometry}

\usepackage[utf8]{inputenc} 
\usepackage[T1]{fontenc}    
\usepackage[colorlinks=true]{hyperref}
\usepackage{url}            
\usepackage{booktabs}       
\usepackage{amsmath,amsfonts,amsthm}
\allowdisplaybreaks  
\usepackage{nicefrac}       
\usepackage{microtype}      

\usepackage{color} 

\usepackage{graphicx}

\usepackage{multirow}
\usepackage{subcaption}
\usepackage{balance}
\usepackage{algorithm}
\usepackage[noend]{algpseudocode}

\usepackage{pdflscape}  

\usepackage{xspace}

\newcommand{\fnref}[1]{\textsuperscript{#1}}

\newcommand{\LTM}{\ensuremath{LTR}\xspace}
\newcommand{\TCT}{\ensuremath{LCF}\xspace}
\newcommand{\LDR}{\ensuremath{LDR}\xspace}
\newcommand{\MOV}{\textsc{MoV}\xspace}
\newcommand{\POV}{\textsc{PoV}\xspace}
\newcommand{\cstar}{c_\star}

\newcommand{\N}{\mathbb{N}}

\newcommand{\Floor}[1]{\left\lfloor #1 \right\rfloor}
\newcommand{\card}[1]{\vert #1 \vert}

\DeclareMathOperator*{\argmax}{arg\,max}
\newcommand{\bmone}{\mathbf{1}}

\usepackage{xparse}
\newcommand{\cond}{\ \middle| \ }
\RenewDocumentCommand\Pr{gg}{%
	\ensuremath{%
		\mathbf{P} \IfNoValueTF{#1}{}{\left( #1 \IfNoValueTF{#2}{}{\cond #2} \right)}
	}
}
\NewDocumentCommand\Prob{mgg}{%
	\ensuremath{%
		\mathbf{P}_{#1} \IfNoValueTF{#2}{}{\left( #2 \IfNoValueTF{#3}{}{\cond #3} \right)}
	}
}
\NewDocumentCommand\Ex{gg}{%
	\ensuremath{%
		\mathbf{E} \IfNoValueTF{#1}{}{\left[ #1 \IfNoValueTF{#2}{}{\cond #2} \right]}
	}
}

\newtheorem{theorem}{Theorem}
\newtheorem{lemma}{Lemma}
\theoremstyle{definition}
\newtheorem{definition}{Definition}

\usepackage{authblk}

\title{
Exploiting Social Influence to Control Elections \\Based on Scoring Rules
}

\author{Federico Cor\`o}
\author{Emilio Cruciani}
\author{Gianlorenzo D'Angelo}
\author{Stefano Ponziani}
\affil{%
 GSSI -- Gran Sasso Science Institute
 \\
 Viale Francesco Crispi 7, 67100, L'Aquila, Italy
 \vspace{0.5em}
 \authorcr
 \texttt{federico.coro@gssi.it}\\
 \texttt{emilio.cruciani@gssi.it}\\
 \texttt{gianlorenzo.dangelo@gssi.it}\\
 \texttt{stefano.ponziani@gssi.it}
}

\date{}

\begin{document}
\maketitle

\begin{abstract}
Online social networks are used to diffuse opinions and ideas among users, enabling a faster communication and a wider audience.
The way in which opinions are conditioned by social interactions is usually called social influence. 
Social influence is extensively used during political campaigns to advertise and support candidates.

Herein we consider the problem of exploiting social influence in a network of voters in order to change their opinion about a target candidate with the aim of increasing his chance to win/lose the election in a wide range of voting systems.
 
We introduce the \emph{Linear Threshold Ranking}, a natural and powerful extension of the well-established \emph{Linear Threshold Model}, which describes the change of opinions taking into account the amount of exercised influence.
We are able to maximize the score of a target candidate up to a factor of $1-1/e$ by showing submodularity.
We exploit such property to provide a $\frac{1}{3}(1-1/e)$-approximation algorithm for the \emph{constructive} election control problem.
Similarly, we get a $\frac{1}{2}(1-1/e)$-approximation ratio in the \emph{destructive} scenario.
The algorithm can be used in \emph{arbitrary scoring rule voting systems}, including \emph{plurality rule} and \emph{borda count}.
Finally, we perform an experimental study on real-world networks, measuring Probability of Victory (\POV) and Margin of Victory (\MOV) of the target candidate, to validate the model and to test the capability of the algorithm.
\end{abstract}

\clearpage

\section{Introduction}
\label{sec:intro}

As humans, we usually have a specific personal opinion on certain topics, such as lifestyle or consumer products.
These opinions, normally formed on personal life experience and information, can be conditioned by the interaction with our friends leading to a change in our original opinion on a particular topic if a large part of our friends holds a different opinion.

In the last decades, this phenomenon of opinion diffusion has been intensely investigated from many different perspectives, from sociology to economics.
In recent years, there has been a growing interest on the relationship between social networks and political campaigning.
Political campaigns nowadays use online social networks to lead elections in their favor; for example, they can target specific voters with fake news~\cite{allcott2017social}.
A real-life example of political intervention in this context occurred in the US Congressional elections in 2010, where a set of users were encouraged to vote with a message on Facebook.
These messages directly influenced the real-world voting behavior of millions of people~\cite{bond201261}.
Another example is that of French elections in 2017, where automated accounts in social networks spread a considerable portion of political content, mostly fake news, trying to influence their outcome~\cite{ferrara2017disinformation}.

There exist an extensive literature on manipulating elections without considering the underlying
social network structure of the voters, e.g., swap bribery~\cite{elkind2009swap}, shift bribery~\cite{bredereck2016complexity}; 
we point the reader to a recent survey~\cite{faliszewski2016control}.
%

The study of opinion diffusion modeled as a majority dynamics has attracted much attention in recent literature \cite{auletta2015minority,brill2016pairwise,botan2017propositionwise}.
In these models each agent has an initial preference list and at each time step a subset of agents updates their opinions, i.e., their preference lists, according to some majority-based rule that depends on the opinions of their neighbors in the network.

Nevertheless, there are only few studies on the opinion diffusion on social networks.
The Independent Cascade Model~\cite{kempe2003maximizing} has been considered as diffusion process to guarantee that a target candidate wins/loses~\cite{bartholdi1992hard,hemaspaandra2007anyone}.
The constructive (destructive) election control problem consists in targeting a specific candidate to change voters' opinions about him with the aim of maximizing (minimizing) his margin and probability of victory~\cite{wilder2018controlling}.
A variant of the Linear Threshold Model~\cite{kempe2003maximizing} with weights on the vertices has been considered on a graph in which each node is a cluster of voters with a specific list of candidates and there is an edge between two nodes if they differ by the ordering of a single pair of adjacent candidates~\cite{faliszewski2018opinion}.
Moreover, it has been studied how to manipulate the network in order to have control on the majority opinion, e.g., bribing or adding/deleting edges, on
a simple Linear Threshold Model where each node holds a binary opinion, each edge has the same fixed weight, and all vertices have a threshold fixed to $1/2$~\cite{bredereck2017manipulating}.

In this work we focus on a variant of the election control through social influence problem defined in~\cite{wilder2018controlling}:
Given a social network of people willing to vote, we want to select a fixed-size subset of voters such that their influence on the others will change the election outcome, maximizing the chances of a target candidate to win or lose, in the specific scenario in which only the opinions about a target candidate can be changed.
Differently from previous work~\cite{wilder2018controlling} we consider more general voting systems and a different diffusion model, that takes into account the degree of influence that voters exercise on the others and is able to describe the scenario in which a high influence on someone can radically change its opinion.
In this setting we prove the nontrivial fact that any scoring function is monotone and submodular with respect to the initial set of active nodes.
Moreover we exploit this fact to prove a constant factor approximation of the election control problem in our model.

\subsection{Original Contribution}
\begin{itemize}
\item We introduce the \emph{Linear Threshold Ranking}, a natural and powerful extension of the \emph{Linear Threshold Model} for the election scenario that takes into account the degree of influence of the voters on each other.
\item We prove that, in such model, we can achieve a $(1-1/e)$-approximation to the problem of maximizing the score of a target candidate by proving submodularity in the general case of the \emph{scoring rule} for \emph{arbitrary} scoring function (including popular voting systems, e.g., \emph{plurality rule} or \emph{borda count}), with any number of candidates.
\item We exploit the $(1-1/e)$-approximation algorithm that maximizes the score to achieve an extra approximation factor of $\frac{1}{3}$ to the problem of maximizing the \emph{Margin of Victory} of a target candidate in \emph{arbitrary scoring rule voting systems} with any number of candidates.
\item We give a simple reduction that maps \emph{destructive control} problems to constructive control ones and allows us to achieve a $\frac{1}{2}(1-1/e)$-approximation to the destructive control problem.
\item We perform simulations of our model on four heterogeneous real-world networks and test the capability of our algorithm.
\end{itemize}

\section{Background}
\label{sec:preliminaries}

In this section we present some notions and concepts about \emph{voting systems} and \emph{influence maximization} on social networks that will be used in the design and analysis of the algorithm.
Moreover we formally introduce the problem and the notation used to analyze it.

\subsection{Voting Systems}
\label{ssec:voting}

\emph{Voting systems} are sets of rules that regulate all aspects of elections and  that determine their outcome.
In particular a voting system decides candidates and voters eligibility, other than fixing the rules for determining the winner of the elections.
\emph{Social choice theory} formally defines and analyzes voting systems, studying how the combination of individual opinions or preferences reaches a collective decision; 
\emph{computational social choice}, instead, studies the computational complexity of outcomes of voting rules and can serve as a barrier against strategic manipulation in elections~\cite{chevaleyre2007short,faliszewski2010ai,brandt2016handbook,endriss2017trends}.

We focus on two \emph{single-winner} voting systems:
\begin{itemize}
    \item \emph{Plurality rule}: Each voter can only express a single preference among the candidates and that with the \emph{plurality} of the votes wins, i.e., it is sufficient to have the highest number of votes and there is no need of an absolute majority (50\%+1 of votes).
    \item \emph{Scoring rule}: Each voter expresses his preference as a \emph{ranking}; each candidate is then assigned a \emph{score}, computed as a function of the positions he was ranked among the voters.
\end{itemize}
The former is arguably the simplest scenario and is one of the most commonly used for national legislatures and presidential elections.
The latter is a very general definition, but can include several popular election methods by choosing an adequate \emph{scoring function}:
\begin{itemize}
    \item if the scoring function assigns 1 point to the first candidate and 0 to all the others this is equivalent to the \emph{plurality rule};
    \item if the scoring function assigns 1 point to the first $t$ candidates and 0 to the others then it is equivalent to the \emph{$t$-approval}, where each voter approves $t$ candidates;
    \item if the scoring function assigns 1 point to the first $m-t$ candidates and 0 to the remaining $t$, where $m$ is the total number of candidates, then it is equivalent to the \emph{$t$-veto} or \emph{anti-plurality} rule;
    \item if the scoring function assigns $m-l$ points to the candidate in position $l$ then it is equivalent to the \emph{borda count}, in which each voter ranks the candidates and each candidate gets a score equal to the number of candidates ranked lower in each list.
\end{itemize}

\subsection{Influence Maximization}
\label{ssec:influence}
The \textit{influence maximization} problem studies a social network represented as a graph and has the goal of finding the $B$-sized set of influential nodes that can maximize the spread of information~\cite{kempe2015maximizing}.
In general, all existing diffusion models can be categorized into three classes:
cascade models, threshold models, and epidemic models. The most popular for studying social influence problems are the \emph{Independent Cascade Model} (ICM) and the \emph{Linear Threshold Model} (LTM). 
These models are graph-based, namely they assume an underlying directed graph where nodes represent agents and edges represent connections between them. Each node can be either \emph{active}, that is it spreads the information, or \emph{inactive}. With some probability, active nodes diffuse the information to their neighbors.
The ICM model requires a diffusion probability to be associated with
each edge, whereas LTM requires an influence degree to be defined on each edge and
an influence threshold on each node. 
For both models, the diffusion process proceeds
iteratively in a synchronous way along a discrete time-axis, starting from an initial set of nodes, usually called \emph{seeds}.

In this work we focus on LTM~\cite{kempe2003maximizing}.
Given a graph $G=(V, E)$, in LTM each node $v \in V$ has a threshold $t_v \in [0,1]$ sampled uniformly at random and independently from the other nodes and each edge $(u,v) \in E$ has a weight $b_{uv} \in [0,1]$ with the constraint that, for each $v \in V$, the sum of the weights of the incoming edges of $v$ is less or equal to 1, i.e., $\sum_{(u,v) \in E} b_{uv} \leq 1$. 
Let $A_0 \subseteq V$ be the set of \emph{active} nodes at the beginning of the process.
More in general, let $A_t \subseteq V$ be the set of nodes active at time $t$.
In LTM an inactive node $v$ becomes active if the sum of the weights of the edges coming from nodes that are active at the previous round is greater than or equal to its threshold $t_v$, i.e., $v \in A_t$ if and only if $v \in A_{t-1}$ or $\sum_{u \in A_{t-1} : (u,v) \in E} b_{uv} \geq t_v$.
When a node is active, it influences its neighbors and increases the chance of making them change their preference list.
The process has \emph{quiesced} at the first time $\tilde{t}$ such that the set of active nodes would not change in the next round, i.e., time $\tilde{t}$ is such that $A_{\tilde{t}} = A_{\tilde{t}+1}$.
We define the eventual set of active nodes as $A := A_{\tilde{t}}$.

The distribution of the set of active nodes in the graph starting with $A_0$ under the LTM process is equivalent to the distribution reachable from the same set $A_0$ in the set of random graphs called \textit{live-edge graphs} (Theorem~\ref{teo:kempe_equivalence}). 
A live-edge graph is built as follows: Given an influence graph $G=(V, E)$, for every node $v \in V$ select at most one of its incoming edges with probability proportional to the weight of that edge, i.e., edge $(u,v)$ is selected with probability $b_{uv}$, and no edge is selected with probability $1 - \sum_{u \in N_v} b_{uv}$.
Let us denote by $\mathcal{G}$ the set of all possible live-edge graphs that can be generated from $G$. 
\begin{theorem}[Kempe, Kleinberg, and Tardos~\cite{kempe2015maximizing}]
Given a graph $G=(V, E)$ and an initial set of nodes $A_0 \subseteq V$, the distribution of the sets of \emph{active} nodes in $G$ after LTM has quiesced starting from $A_0$ is equal to the distribution of the sets of nodes that are \emph{reachable} from $A_0$ in the set of live-edge graphs $\mathcal{G}$.
\label{teo:kempe_equivalence}
\end{theorem}
Moreover, under the live-edge model, the problem of selecting the initial set of nodes in order to maximize the diffusion is \emph{submodular}~\cite{kempe2015maximizing}.
Therefore, exploiting a classical result~\cite{NWF78}, the influence maximization problem can be approximated to a constant factor of $1 - 1/e$ using a simple greedy hill-climbing
approach that starts with an empty solution and, for $B$ iterations, selects a single node that gives the maximal marginal gain on the objective function with respect the solution computed so far.
This algorithm guarantees the best approximation, but is still very computational expensive: Evaluating the expected number of active nodes is \mbox{$\mathit{\#P}$-hard}~\cite{chen2010scalable}.
There exists a simulation-based approach in which a Monte-Carlo simulation is performed to evaluate the influence spread of a given seed set $A$~\cite{kempe2015maximizing}.
The standard Chernoff-Hoeffding bounds imply $1\pm\epsilon$ approximation to the expected number of active nodes by simulating a polynomial number of times the diffusion process.

\section{Linear Threshold Ranking and Election Control}
\label{sec:problem}

We consider the scenario in which a set of candidates are running for the elections and a social network of voters will decide the winner. 
In particular we focus on the simple \emph{plurality rule} and on the more general case of the \emph{scoring rule}.

Some attacker could be interested in changing the outcome of the elections by targeting a subset of voters with advertisement or (possibly fake) news about one specific candidate.
Such voters, with some probability, can influence their friends by sharing the news.
Suppose the attacker has a budget that can use to target some voters and that they will start a diffusion process that changes opinions in the network.
Is it possible for the attacker to select a set of voters in order to have constructive/destructive control over a target candidate, i.e., to change the voters' opinions on this candidate in order to maximizes his chances to win/lose the elections?

More formally, let $G = (V, E)$ be a directed graph representing the underlying social network of people willing to vote. 
For each node $v \in V$ we call $N_v$ the sets of incoming neighbors of $v$.
Let $C = \{c_1, \ldots, c_m\}$ be a set of $m$ candidates nominated for the elections; we refer to our \emph{target candidate}, i.e., the one that we want to make win/lose the elections, as $\cstar$.
Each $v \in V$ has a permutation $\pi_v$ of $C$, i.e., its list of preferences for the elections;
we denote the position of candidate $c_i$ in the preference list of node $v$ as $\pi_v(c_i)$.

We consider the LTM process starting from an initial set of active nodes $A_0 \subseteq V$. 
Recall that, according to LTM, each node $v \in V$ has a threshold $t_v$, each edge $(u,v) \in E$ has a weight $b_{uv}$, and that $A \subseteq V$ is the set of active nodes at the end of the process.

Let $B \in \N$ be an initial budget that can be used to select the nodes in $A_0$, i.e., the set of active nodes from which the LTM process starts.
In particular, the budget constrains the size of $A_0$, namely $|A_0| \leq B$.

After the LTM process has quiesced, the position of $\cstar$ in the preference list of each node changes according to a function of its incoming active neighbors.
The threshold $t_v$ of each node $v \in V$ models its strength in retaining its original opinion about candidate $\cstar$:
The higher is the threshold $t_v$ the lower is the probability that $v$ is influenced by its neighbors.
Moreover the weight on an edge $b_{uv}$ measures the influence that node $u$ has on node $v$.
Taking into account the role of such parameters, we define the number of positions that $\cstar$ goes up in $\pi_v$ as 
\[
\pi^\uparrow_v(\cstar) :=
\min \left(
\pi_v(\cstar) - 1,
\,
\Floor{
    \frac{\alpha(\pi_v(\cstar))}{t_v} \sum_{\substack{u \in A,\,(u,v) \in E}} b_{uv}
}
\right),
\]
where $\alpha : \{1,\ldots,m\} \rightarrow [0,1]$ is a function that depends on the position of $\cstar$ in $\pi_v$ and models the rate at which $\cstar$ shifts up.
Note that $\alpha$ can be set arbitrarily to model different scenarios, e.g., shifting up of one position from the bottom of the list could be easier than going from the second position to the first with a suitable choice of $\alpha$.
As for $\pi^\uparrow_v(\cstar)$, it can be any integer value in $\{0,\ldots,\pi_v(\cstar)-1\}$:
The floor function guarantees a positive integer value; the minimum between such value and $\pi_v(\cstar)-1$ guarantees that final position of $\cstar$ is at least 1, since the floor function could output too high values when the threshold is small w.r.t.\ the neighbors' influence.
We call this process the \emph{Linear Threshold Ranking} (\LTM).

After the modification of the lists at the end of \LTM, the candidates might have a new position in the preference list of each node $v \in V$; we denote such new preference list as $\tilde{\pi}$.
In particular, the new position of candidate $\cstar$ will be $\tilde{\pi}_v(\cstar) := \pi_v(\cstar) - \pi^\uparrow_v(\cstar)$; 
the candidates that are overtaken by $\cstar$ will shift one position down.

In the problem of \emph{election control} we want to maximize the chances of the target candidate to win the elections under \LTM.
To achieve that, we maximize its expected \emph{Margin of Victory} (\MOV) w.r.t.\ the most voted opponent, akin to that defined in~\cite{wilder2018controlling}.%
\footnote{We actually study the \emph{change} in the margin -- not just the margin -- to have well defined approximation ratios also when the margin is negative.}
Let us consider the general case of the \emph{scoring rule}, where a \emph{nonincreasing} \emph{scoring function} $f : \{1,\ldots,m\} \rightarrow \N$ assigns a score to each position.
Let $c$ and $\tilde{c}$ respectively be the candidates, different from $\cstar$, with the highest score before and after \LTM.
Let
\begin{align}\label{eq:mu0}
\mu(\emptyset) &:= \sum_{v\in V} f(\pi_v(c)) - f(\pi_v(\cstar))
\\\label{eq:muA}
\mu(A_0) &:= \sum_{v\in V} f(\tilde{\pi}_v(\tilde{c})) - f(\tilde{\pi}_v(\cstar))
\end{align}
be the \emph{margin} (i.e., difference in score) between the most voted opponent and $\cstar$ before and after \LTM (Equations~\eqref{eq:mu0} and~\eqref{eq:muA}). 
Thus, the \emph{election control} problem is formalized as that of finding a set of nodes $A_0$ such that
\[
\begin{array}{rl}
    \max_{A_0} & \Ex{\MOV(A_0)} := \Ex{\mu(\emptyset) - \mu(A_0)} \\
    \text{s.t.} & |A_0| \leq B,
\end{array}
\]
namely to find an initial set of seed nodes of at most size $B$ that maximizes the expected \MOV, i.e., change in margin.%
\footnote{Note that \MOV is always positive since the scoring function $f$ is nonincreasing.}

\begin{algorithm}[t]
\caption{\textsc{Greedy}}\label{alg:greedy}
\begin{algorithmic}[1]
\Require{Social graph $G=(V,E)$; Budget $B$; Score function $F$}
\State $A_0 = \emptyset$
\While{$|A_0| \leq B$}
\State $v = \argmax_{w \in V \setminus A_0} F(A_0 \cup \{w\}) - F(A_0)$ 
\State $A_0 = A_0 \cup \{ v \}$
\EndWhile
\State \Return{$A_0$}
\end{algorithmic}
\end{algorithm}

To solve the problem we focus on the score of the target candidate.
Let us define
\begin{align}\label{eq:F0}
F(\emptyset) &:= \sum_{v\in V}f(\pi_v(\cstar))
\\\label{eq:FA}
F(A_0) &:= \Ex{\sum_{v\in V}f(\tilde{\pi}_v(\cstar))}
\end{align}
as the total expected score obtained by candidate $c_\star$ before and after \LTM (Equations~\eqref{eq:F0} and~\eqref{eq:FA}).
In Sections~\ref{sec:plurality} and~\ref{sec:scoring} we prove that the score of the target candidate is a monotone submodular function w.r.t.\ the initial set of seed nodes $A_0$ in both the \emph{plurality} and the \emph{scoring} rule; 
this allows us to get a $(1-1/e)$-approximation of the maximum score through the use of \textsc{Greedy} (Algorithm~\ref{alg:greedy}).
Note that maximizing the score of the target candidate is a $\mathit{NP}$-hard problem:
Consider the case in which there are only two candidates, $\alpha(1) = \alpha(2) = 1$, all nodes have $\cstar$ as second preference, and the scoring function is that of the plurality rule; 
maximizing the score is equal to maximizing the number of active nodes in LTM because when a node becomes active the target candidate shifts of at least one position up (in this case, in first position); thus the two problems are equivalent.
Since influence maximization in LTM is $\mathit{NP}$-hard, then also maximizing the score in \LTM is $\mathit{NP}$-hard because it is a generalization of LTM. Moreover, in this instance, the maximum value of \MOV is equal to twice the maximum score; then the problem of maximizing \MOV is also $\mathit{NP}$-hard.

Although maximizing the score is not equivalent to maximizing \MOV, in Section~\ref{sec:approx} we show that it gives a constant factor approximation to \MOV.
Finally, in Section~\ref{sec:destructive}, we consider the problem of \emph{destructive control}, in which we want the target candidate to lose the elections. 
We prove a constant factor approximation to \MOV also in this case by exploiting a simple reduction that maps it to the constructive case.

\section{Maximizing the Score: Plurality Rule}
\label{sec:plurality}
As a warm-up, in this section we focus on the \emph{plurality rule}. 
We give an algorithm to select an initial set of seed nodes to maximize the expected number of nodes that will change their opinion and have $\cstar$ as first preference at the end of \LTM.

Given a set of initially active nodes $A_0$, let $A$ be the set of nodes that are active at the end of the process. 
An active node $v$ with $\pi_v(\cstar)>1$ will have $\cstar$ as first preference if $\pi^\uparrow_v(\cstar) = \pi_v(\cstar)-1$, that is if and only if 
\[
    \frac{\alpha(\pi_v(\cstar))}{t_v} \sum_{u \in A\cap N_v} b_{uv} \geq \pi_v(\cstar)-1
\]
or, equivalently,
\[ 
t_v \leq \frac{\alpha(\pi_v(\cstar))}{\pi_v(\cstar) - 1} \sum_{u \in A\cap N_v} b_{uv}.
\]

As for the influence maximization problem, we define an alternative random process based on live-edge graphs. 
One possibility could be the following: 
For each live-edge graph evaluate which active nodes satisfy the above formula; 
however, in the live-edge graph process, we don't know the value of $t_v$ since they are sampled uniformly at random at the beginning of LTM. 
To overcome this limitation we introduce a new process, \emph{Live-edge Coin Flip} (\TCT).
\begin{definition}\label{def:lcf}
(\emph{Live-edge Coin Flip process})
\begin{enumerate}
    \item Each node $v \in V$ selects at most one of its incoming edges with probability proportional to the weight of that edge, i.e., edge $(u,v)$ is selected with probability $b_{uv}$, and no edge is selected with probability $1 - \sum_{u \in N_v} b_{uv}$.
    \item Each node $v$ with $\pi_v(\cstar)>1$ that is reachable from $A_0$ in the live-edge graph flips a biased coin and changes its list according to the outcome.
    This is equivalent of picking a random real number $s_v \in [0,1]$ and setting the position of $\cstar$ according to $s_v$ as follows:
    If $s_v \leq \frac{\alpha(\pi_v(\cstar))}{\pi_v(\cstar) - 1}$, node $v$ chooses $\cstar$ as its first preference (i.e., it sets $\tilde{\pi}_v(\cstar)=1$ and shifts all the other candidates down by one position); 
    otherwise, $v$ maintains its original ranking.
\end{enumerate}
\end{definition}
In the following we show that the two processes are equivalent, i.e., starting from any initial set $A_0$ each node in the network has the same probability to end up with $\cstar$ in first position in both processes. 
This allows us to compute the function $F(A_0)$, for a given $A_0$, by solving a reachability problem in graphs, as we will show later in this section. 
We first prove the next Lemma which will be used to show the equivalence between the two processes and to compute $F(A_0)$. 
The lemma shows how to compute the probability that a node $v$ is reachable from $A_0$ at the end of the \TCT process by using the live-edge graphs or by using the probability of the incoming neighbors of $v$ to be reachable from $A_0$.

We denote by $\mathcal{G}$ the set of all possible live-edge graphs sampled from $G$. 
For every $G'=(V,E') \in \mathcal{G}$ we denote by $\Pr{G'}$ the probability that the live edge graph is sampled, namely
\[
\Pr{G'} = \prod_{v : (u,v) \in E'} b_{uv}
\prod_{v : \not\exists (u,v) \in E'} \left( 1 - \sum_{w : (w,v) \in E} b_{wv} \right).
\]
We denote by $R(A_0)$ the set of nodes reachable from $A_0$ at the end of the \TCT process and by $R_{G'}(A_0)$ the set of nodes reachable from $A_0$ in a fixed live-edge graph $G'$ and by $\bmone_{(G', v)}$ the indicator function that is 1 if $v\in R_{G'}(A_0)$ and 0 otherwise.

\begin{lemma}\label{lem:reachability}
Given a set of initially active nodes $A_0$, let $R(A_0)$ be the set of nodes reachable from $A_0$ at the end of the \TCT process. Then 
\begin{align*}
\Pr{v\in R(A_0)} = \sum_{G' \in \mathcal{G}} \Pr{G'}\cdot \bmone_{(G', v)} = \sum_{U \subseteq N_v} \sum_{u \in U} b_{uv} \cdot \Pr{( R(A_0) \cap N_v) = U}.
\end{align*}
\end{lemma}
\begin{proof}
By the law of total probability
\[
\Pr{v\in R(A_0)} =  \sum_{G' \in \mathcal{G}}\Pr{v\in R(A_0)}{G'}\cdot\Pr{G'}.
\]
Given a live-edge graph $G'$ sampled form $\mathcal{G}$, the value of $\Pr{v\in R}{G'}$ is equal to 1 if $v$ is reachable from $A_0$ in $G'$, and it is 0 otherwise. Then
\[
\sum_{G' \in \mathcal{G}}\Pr{v\in R(A_0)}{G'}\cdot\Pr{G'} = \sum_{G' \in \mathcal{G}} \Pr{G'}\cdot \bmone_{(G', v)}.
\]
which shows the first part of the lemma. We now show the following equality:
\begin{equation}\label{eq:plurality:one}
\sum_{G' \in \mathcal{G}} \Pr{G'}\cdot \bmone_{(G', v)}  = \sum_{U \subseteq N_v} \sum_{u \in U} b_{uv} \cdot \Pr{( R(A_0) \cap N_v) = U}.
\end{equation}
We can re-write the left hand side as
\begin{equation*}
\sum_{G' \in \mathcal{G}} \Pr{G'} \cdot \bmone_{(G', v)} = \sum_{U \subseteq N_v} \sum_{\substack{G' \in \mathcal{G} \text{ s.t.}\\ R_{G'}(A_0)\cap N_v =U}} \Pr{G'} \cdot \bmone_{(G', v)}.
\end{equation*}

In each live-edge graph $G'$ for which $\Pr{G'} \cdot \bmone_{(G', v)} \neq 0$ node $v$ selected one of its incoming edges and then
\(\Pr{v \mbox{ selected } u \mbox{ in } \TCT} = b_{uv},\) 
for each $u\in N_v$. 
Therefore,  the above value is equal to
\begin{align*}
&\sum_{U \subseteq N_v} \sum_{\substack{G' \in \mathcal{G} \text{ s.t.}\\ R_{G'}(A_0)\cap N_v =U}} \sum_{u\in U} \Pr{G'}{v \mbox{ selected $u$ in } \TCT }b_{uv}
\\
&=\sum_{U \subseteq N_v} \sum_{u\in U} b_{uv}\sum_{\substack{G' \in \mathcal{G} \text{ s.t.}\\ R_{G'}(A_0)\cap N_v =U}}  \Pr{G'}{v \mbox{ selected $u$ in } \TCT },
\end{align*}
where the first equality is due to the law of total probability and the last one is just reordering of the terms of the sums.

In each live-edge $G'$ that does not contain the edge $(u,v)$, the probability \(\Pr{G'}{v \mbox{ selected $u$ in } \TCT }\) is equal to zero. Then,
\begin{align*}
&\sum_{\substack{G' \in \mathcal{G} \text{ s.t.}\\ R_{G'}(A_0)\cap N_v =U}}  \Pr{G'}{v \mbox{ selected $u$ in } \TCT }
\\
&=\sum_{\substack{G' \in \mathcal{G} \text{ s.t.}\\ R_{G'}(A_0) \cap N_v =U\\(u,v)\in E'}}  \Pr{G'}{v \mbox{ selected $u$ in } \TCT }.
\end{align*}
By definition of conditional probability we have that the above sum is equal to:
\begin{equation*}
\sum_{\substack{G' \in \mathcal{G} \text{ s.t.}\\ R_{G'}(A_0) \cap N_v =U\\(u,v)\in E'}}  \frac{\Pr{G'\cap (v \mbox{ selected $u$ in } \TCT) }}{b_{uv}}.
\end{equation*}
Since, in each $G'$ considered in the sum, edge $(u,v)$ belongs to $G'$, this is equal to:
\begin{equation}\label{eq:plurality:two}
\sum_{\substack{G' \in \mathcal{G} \text{ s.t.}\\ R_{G'}(A_0) \cap N_v =U\\(u,v)\in E'}}  \frac{\Pr{G'}}{b_{uv}}.
\end{equation}

Let us now consider the right hand side of Equality~\eqref{eq:plurality:one}. For each $U\subseteq N_v$, the probability that $( R(A_0) \cap N_v) = U$ is given by the sum of the probabilities of all the live-edge graphs that satisfy this property, since these graphs represent disjoint events, we have:
\begin{equation*}
\Pr{( R(A_0) \cap N_v) = U} = \sum_{\substack{G' \in \mathcal{G} \text{ s.t.}\\ R_{G'}(A_0) \cap N_v =U}}  \Pr{G'}.
\end{equation*}
Let us fix a node $u\in U$. For each $G' \in \mathcal{G}$ such that $(R_{G'}(A_0)\cap N_v) = U$, there exists a live-edge graph $G''$ that has the same edges as $G'$ but has edge $(u,v)$ as incoming edge of $v$. Since all the other edges of $G''$ are equal to those of $G'$, then $(R_{G''}(A_0) \cap N_v) =U$.

\noindent We have that 
\begin{equation*}
\Pr{G'} = \left\{
\begin{array}{ll}\displaystyle
\frac{\Pr{G''}}{b_{uv}}\cdot b_{u_i v} &\mbox{if } \exists (u_i, v) \text{ in } G',
\\\displaystyle
\frac{\Pr{G''}}{b_{uv}}\cdot \left(1- \sum_{u_i\in N_v}b_{u_i v}\right)&\text{otherwise.}
\end{array}
\right.
\end{equation*}
Therefore,
\small{
\begin{align*}
\sum_{\substack{G' \in \mathcal{G} \text{ s.t.}\\ R_{G'}(A_0) \cap N_v =U}}  \Pr{G'} 
&=\sum_{\substack{G'' \in \mathcal{G} \text{ s.t.}\\ R_{G''}(A_0) \cap N_v =U\\(u,v)\in E'}}
\left(  \sum_{u_i\in N_v}b_{u_iv} \frac{\Pr{G''}}{b_{uv}}
+ \left(1-\sum_{u_i\in N_v}b_{u_iv}\right)\frac{\Pr{G''}}{b_{uv}}
\right)
\\
&=\sum_{\substack{G'' \in \mathcal{G} \text{ s.t.}\\ R_{G''}(A_0) \cap N_v =U\\(u,v)\in E'}}  \frac{\Pr{G''}}{b_{uv}}.
\end{align*}}
Equality~\eqref{eq:plurality:one} follows since the above expression is equal to~\eqref{eq:plurality:two}.
\end{proof}

The next theorem shows the equivalence between $\LTM$ and $\TCT$.
\begin{theorem}\label{th:plurality}
Given a set of initially active nodes $A_0$, let $A'_{\LTM}$ and $A'_{\TCT}$ be the set of nodes such that $\tilde{\pi}_v(\cstar) = 1$ at the end of \LTM and \TCT, respectively, both starting from $A_0$. 
Then, for each $v\in V$, $\Pr{v\in A'_{\LTM}} = \Pr{v\in A'_{\TCT}}$.
\end{theorem}
\begin{proof}
We exclude from the analysis nodes $v$ with $\pi_v(\cstar)=1$ since they keep their original ranking in both models.
Let us start by analyzing the \LTM process.
Let $A$ be the set of active nodes at the end of the \LTM process that starts from $A_0$.
If $U$ is the maximal subset of active neighbors of $v$ (i.e. $U =A\cap N_v$), then we can write the probability that $v\in A_{\LTM}'$ given $U$, as
\begin{align*}
\Pr{v \in A_{\LTM}'}{(A\cap N_v) = U} 
&= \Pr{t_v \leq \frac{\alpha(\pi_v(\cstar))}{\pi_v(\cstar) - 1} \sum_{u \in U} b_{uv}} 
\\
&= \frac{\alpha(\pi_v(\cstar))}{\pi_v(\cstar) - 1} \sum_{u \in U} b_{uv}.
\end{align*}
The overall probability that $v\in  A_{\LTM}'$ is
\begin{align*}
\Pr{v \in A_{\LTM}'} 
&= \sum_{U\subseteq N_v} \Pr{v \in A_{\LTM}'}{(A\cap N_v) = U}\cdot\Pr{U =(A\cap N_v)}
\\
&= \frac{\alpha(\pi_v(\cstar))}{\pi_v(\cstar) - 1} \sum_{U\subseteq N_v} \sum_{u \in U} b_{uv} \cdot \Pr{(A\cap N_v) = U}.
\end{align*}

Let us now analyze the \TCT process. In order for $v$ to be in $A'_{\TCT}$ it must hold that the coin toss has a positive outcome and that $v\in R$. Thus,
\begin{equation}\label{eq:plurality:three}
    \Pr{v\in A'_{\TCT}} = \frac{\alpha(\pi_v(\cstar))}{\pi_v(\cstar) - 1} \Pr{v\in R(A_0)}.
\end{equation}
By Lemma~\ref{lem:reachability}, we have
\[
\Pr{v\in A'_{\TCT}} = \frac{\alpha(\pi_v(\cstar))}{\pi_v(\cstar) - 1}  \sum_{U \subseteq N_v} \sum_{u\in U} b_{uv}\cdot\Pr{(R(A_0) \cap N_v) = U}.
\]
By Theorem~\ref{teo:kempe_equivalence}, $\Pr{(R(A_0) \cap N_v) = U}= \Pr{(A\cap N_v) = U}$, and hence the theorem follows.
\end{proof}
We now exploit Theorem~\ref{th:plurality} to show how to compute the value of $F(A_0)$. For each positive integer $r\leq m$, we denote by $V_{c_i}^r$ the set of nodes that have candidate $c_i$ in position $r$. In the case of plurality rule, $F(A_0)$ is the expected cardinality of $A'_{\LTM}$, that is
\[
  F(A_0) = \Ex{\card{A'_{\LTM}}} = \sum_{v \in V}\Pr{v\in A'_{\LTM}}.
\]
By Theorem~\ref{th:plurality} and Equality~\eqref{eq:plurality:three}, this is equal to
\[
\sum_{v \in V}\Pr{v\in A'_{\TCT}} = F(\emptyset) +  \sum_{v \in V,\,\pi_v(\cstar)>1} \frac{\alpha(\pi_v(\cstar))}{\pi_v(\cstar) - 1} \Pr{v\in R(A_0)}.
\]
By Lemma~\ref{lem:reachability}, it follows that 
\[
F(A_0) =F(\emptyset) + \sum_{v \in V,\pi_v(\cstar)>1} \frac{\alpha(\pi_v(\cstar))}{\pi_v(\cstar) - 1} \sum_{G' \in \mathcal{G}} \Pr{G'} \cdot\bmone_{(G',v)}.
\]
We can rewrite the above formula as follows:
\begin{align*}
F(A_0) -F(\emptyset) 
&=\sum_{r=2}^{m} \sum_{v: \pi_v(\cstar)=r} \frac{\alpha(r)}{r - 1} \sum_{G' \in \mathcal{G}} \Pr{G'} \cdot\bmone_{(G',v)}
\\
&=\sum_{r=2}^{m} \frac{\alpha(r)}{r - 1} \sum_{G' \in \mathcal{G}} \Pr{G'} \sum_{v: \pi_{v}(\cstar)=r} \bmone_{(G',v)}
\\
&= \sum_{r=2}^{m} \frac{\alpha(r)}{r - 1} \sum_{G' \in \mathcal{G}} \Pr{G'} \cdot \card{\{v : v \in R_{G'}(A_0) \wedge \pi_v(\cstar) = r\}}
\\
&= \sum_{r=2}^{m} \frac{\alpha(r)}{r - 1} \sum_{G' \in \mathcal{G}} \Pr{G'} |R_{G'}(A_0, V^r _{\cstar})|,
\end{align*}
where, for a graph $G'\in \mathcal{G}$ and a positive integer $r\leq m$, we denoted by $R_{G'}(A_0, V^r _{\cstar})$ the subset of $V^r _{\cstar}$ of nodes reachable from a set of nodes $A_0$ in $G'$, $R_{G'}(A_0, V^r _{\cstar}) = \{v : v \in R_{G'}(A_0)\wedge \pi_v(\cstar) = r\}$.

It follows that the function $F(A_0)$ is a non-negative linear combination of functions $|R_{G'}(A_0, V^r _{\cstar})|$. In the next lemma, we show that $R_{G'}(A_0, V^r _{\cstar})$ in $G'$ is a monotone submodular%
\footnote{For a ground set $N$, a function $z:2^N\rightarrow \mathbb{R}$ is \emph{submodular} if for any two sets $S,T$ such that $S\subseteq T \subseteq N$ and for any element $e\in N\setminus T$ it holds that $z(S\cup\{e\}) - z(S) \geq z(T\cup \{e\}) - z(T)$.}
function of the initial set of nodes $A_0$. This implies that also $F(A_0)$ is monotone and submodular w.r.t.\ $A_0$ and the same holds for $F(A_0) - F(\emptyset)$. Therefore, we can use \textsc{Greedy} (Algorithm~\ref{alg:greedy}) to find a set $A_0$ whose value $F(A_0) - F(\emptyset)$ is at least $1-1/e$ times the one of an optimal solution for election control problem~\cite{NWF78}.
Note that, we can use the same algorithm to approximate $F(A_0)$ within the same approximation bound.

\begin{lemma}\label{lem:submodularity}
Given a graph $G'\in \mathcal{G}$ and a positive integer $r\leq m$ , the size of $R_{G'}(A_0, V^r_{\cstar})$ in $G'$ is a monotone submodular function of the initial set of nodes $A_0$.
\end{lemma}
\begin{proof}
Given $A_0\subseteq V$, for any $v\in V\setminus A_0$, the nodes in $V^r _{\cstar}$ that are reachable from $A_0$ in $G'$ are reachable also from $A_0\cup \{v\}$. Therefore, 
\[
 |R_{G'}(A_0\cup \{v\}, V^r _{\cstar})|\geq |R_{G'}(A_0, V^r _{\cstar})|.
\]

Let us consider two sets of nodes $S,T$ such that $S\subseteq T \subseteq V$ and a node $v\in V\setminus T$. We show that 
$|R_{G'}(S\cup \{v\}, V^r _{\cstar})| - |R_{G'}(S, V^r _{\cstar})| \geq 
|R_{G'}(T\cup \{v\}, V^r _{\cstar})| - |R_{G'}(T, V^r _{\cstar})|$.
Since $v\in S\cup\{v\}$, we have that
\[
|R_{G'}(S\cup \{v\}, V^r _{\cstar})| - |R_{G'}(S, V^r _{\cstar})| = 
|R_{G'}(S\cup \{v\}, V^r _{\cstar}) \setminus R_{G'}(S, V^r _{\cstar})|.
\]
Moreover, for any two sets of nodes $B,C$ we have that $R_{G'}(B\cup C, V^r _{\cstar}) = R_{G'}(B, V^r _{\cstar}) \cup R_{G'}(C, V^r _{\cstar})$.
Hence
\begin{align*}
 R_{G'}(S\cup \{v\}, V^r _{\cstar})\setminus R_{G'}(S, V^r _{\cstar})
 &= [ R_{G'}(S, V^r _{\cstar})\cup R_{G'}(\{v\}, V^r _{\cstar}) ]\setminus R_{G'}(S, V^r _{\cstar})
 \\
 &= R_{G'}(\{v\}, V^r _{\cstar})\setminus R_{G'}(S, V^r _{\cstar}).
\end{align*}
Similarly,
\[
|R_{G'}(T\cup \{v\}, V^r _{\cstar})| - |R_{G'}(T, V^r _{\cstar})| = |R_{G'}(\{v\}, V^r _{\cstar})\setminus R_{G'}(T, V^r _{\cstar})|.
\]
Since $S\subseteq T$, then $R_{G'}(S, V^r _{\cstar})\subseteq R_{G'}(T, V^r _{\cstar})$ and then $R_{G'}(\{v\}, V^r _{\cstar}) \setminus R_{G'}(S, V^r _{\cstar}) \supseteq R_{G'}(\{v\}, V^r _{\cstar})  \setminus R_{G'}(T, V^r _{\cstar})$, which implies the statement.
\end{proof}

\section{Maximizing the Score: Scoring Rule}
\label{sec:scoring}
In this section we extend the results of Section~\ref{sec:plurality} to the general case of the \emph{scoring rule}, in which a \emph{scoring function} $f$ assigns a score to each candidate according to the positions he was ranked in the voters' lists. 
The overall approach is similar, but more general: 
We first define an alternative random process to \LTM and show its equivalence to \LTM; 
then we use this model to compute $F(A_0)$ and show that it is a monotone submodular function of the initial set of active nodes $A_0$. 
This latter result allows us to compute a set $A_0$ that has an approximation guarantee of $1-1/e$ on the maximization of the score of the target candidate.

The alternative random process, called \emph{Live-edge Dice Roll} (\LDR), is defined as follows.
\begin{definition}\label{def:ldr}
(\emph{Live-edge Dice Roll process})
\begin{enumerate}
    \item Each node $v \in V$ selects at most one of its incoming edges with probability proportional to the weight of that edge, i.e., edge $(u,v)$ is selected with probability $b_{uv}$, and no edge is selected with probability $1 - \sum_{u \in N_v} b_{uv}$.
    \item Each node $v$ with $\pi_v(\cstar)>1$ that is reachable from $A_0$ in the live-edge graph rolls a biased $\pi_v(\cstar)$-sided dice and changes its list according to the outcome. 
    This is equivalent to picking a random real number $s_v$ in $[0,1]$ and setting the position of $\cstar$ according to $s_v$ as follows:
    \[
    \tilde{\pi}_v(\cstar) = \left\{
    \begin{array}{ll}
         1    & \text{if } s_v \leq \frac{\alpha(\pi_v(\cstar))}{\pi_v(\cstar)-1},\\
         \ell & \text{if }  \frac{\alpha(\pi_v(\cstar))}{\pi_v(\cstar)-\ell+1} < s_v \leq \frac{\alpha(\pi_v(\cstar))}{\pi_v(\cstar)-\ell},\\
         &\quad \text{for }\ell = 2,\ldots, \pi_v(\cstar)-1,\\
         \pi_v(\cstar)  & \text{if }  s_v > \alpha(\pi_v(\cstar)).
    \end{array}
    \right.
    \]
    If $\tilde{\pi}_v(\cstar)\neq \pi_v(\cstar)$, all candidates between $\tilde{\pi}_v(\cstar)$ and $\pi_v(\cstar)-1$ are shifted down by one position.
\end{enumerate}
\end{definition}

In the next theorem we show that processes \LTM and \LDR have the same distribution.
\begin{theorem}\label{th:scoring}
Given a set of initially active nodes $A_0$ and a node $v\in V$, let $\tilde{\pi}_v^{\LTM}(\cstar)$ and $\tilde{\pi}_v^{\LDR}(\cstar)$ be the position of node $v$ at the end of \LTM and \LDR, respectively, both starting from $A_0$. Then, $\Pr{\tilde{\pi}_v^{\LTM}(\cstar)=\ell} = \Pr{\tilde{\pi}_v^{\LDR}(\cstar)=\ell}$, for each $\ell=1,\ldots, \pi_v(\cstar)$.
\end{theorem}
\begin{proof}
Let $A$ be the set of active nodes at the end of the \LTM process that starts from $A_0$.
The probability that an active node moves candidate $\cstar$ to position $\ell$ is given by the following function.
\begin{definition}\label{def:P(r,l)}
For each $r,\ell\in \{1,\ldots,m\}$ we define:
\[
    \Pr{r,\ell} =\left\{\begin{array}{ll}
         \frac{\alpha(r)}{r-1} & \mbox{if } \ell=1,\\
         \frac{\alpha(r)}{r-\ell} - \frac{\alpha(r)}{r-\ell+1} & \mbox{if } \ell=2,\ldots, r-1,\\
         1- \alpha(r) & \mbox{if }  \ell=r.
    \end{array}
    \right.
\]
\end{definition}
In particular, for a node $v$, the probability that the second step of \LDR yields $\tilde{\pi}_v(\cstar) = \ell$, for $\ell=1,\ldots, \pi_v(\cstar)$, is $\Pr(\pi_v(\cstar),\ell)$.
Thus,
\[
    \Pr{\tilde{\pi}_v^{\LTM}(\cstar)=\ell}
    = \sum_{U\subseteq N_v} \Pr{\tilde{\pi}_v^{\LTM}(\cstar)=\ell}{(A\cap N_v) = U}\cdot\Pr{(A\cap N_v) = U}.
\]
If $U$ is is the maximal subset of active neighbors of $v$ (i.e., $U =A\cap N_v$), then we can write the probability that $\tilde{\pi}_v^{\LTM}(\cstar)=\ell$ given $U$ as follows:
\[
\Pr{\tilde{\pi}_v^{\LTM}(\cstar)=\ell}{(A\cap N_v) = U} 
=\Pr{t_v \leq \frac{\alpha(\pi_v(\cstar))}{\pi_v(\cstar) - 1} \sum_{u \in U} b_{uv}}
\]
if $\ell=1$;
\begin{align*}
&\Pr{\tilde{\pi}_v^{\LTM}(\cstar)=\ell}{(A\cap N_v) = U}
\\
&= \Pr{\frac{\alpha(\pi_v(\cstar))}{\pi_v(\cstar) - \ell + 1} \sum_{u \in U} b_{uv}<t_v \leq \frac{\alpha(\pi_v(\cstar))}{\pi_v(\cstar) - \ell} \sum_{u \in U} b_{uv}}
\end{align*}
if $\ell=2,\ldots, \pi_v(\cstar)-1$;
\[
\Pr{\tilde{\pi}_v^{\LTM}(\cstar)=\ell}{(A\cap N_v) = U} = \Pr{t_v > \alpha(\pi_v(\cstar)) \sum_{u \in U} b_{uv}}
\]
if $\ell=\pi_v(\cstar)$. In other words,
\[
\Pr{\tilde{\pi}_v^{\LTM}(\cstar)=\ell}{(A\cap N_v) = U} = \Pr(r,\ell)\sum_{u \in U} b_{uv}.
\]
 Therefore,
\[
\Pr{\tilde{\pi}_v^{\LTM}(\cstar)=\ell} = \Pr(\pi_v(\cstar),\ell)\sum_{U\subseteq N_v}\sum_{u \in U} b_{uv} \Pr{(A\cap N_v) = U}.
\]

In \LDR, $\Pr{\tilde{\pi}_v^{\LDR}(\cstar)=\ell} =\Pr(v\in R(A_0))\cdot \Pr(\pi_v(\cstar),\ell)$. By Lemma~\ref{lem:reachability}, it follows that 
\[
    \Pr{\tilde{\pi}_v^{\LDR}(\cstar)=\ell}
    = \Pr(\pi_v(\cstar),\ell)\sum_{U\subseteq N_v}\sum_{u \in U} b_{uv} \Pr{(R(A_0)\cap N_v) = U}.
\]
By Theorem~\ref{teo:kempe_equivalence}, $\Pr{(R(A_0)\cap N_v) = U} = \Pr{(A\cap N_v) = U}$, which shows the statement.
\end{proof}
By definition, the value of $F(A_0)$ is
\[
F(A_0) = \Ex{\sum_{v\in V}f(\tilde{\pi}_v(\cstar))}= 
\sum_{v\in V}\sum_{\ell = 1}^{\pi_v(\cstar)} f(\ell)\Pr{\tilde{\pi}_v(\cstar) = \ell}.
\]
In \LDR, $\Pr{\tilde{\pi}_v^{\LDR}(\cstar)=\ell} =\Pr(v\in R(A_0))\cdot \Pr(\pi_v(\cstar),\ell)$, moreover, by Lemma~\ref{lem:reachability}, $\Pr(v\in R(A_0)) =\sum_{G' \in \mathcal{G}} \Pr{G'} \bmone_{(G',v)}$. Then,
\[
F(A_0) =\sum_{v\in V}\sum_{\ell = 1}^{\pi_v(\cstar)} f(\ell)\Pr(\pi_v(\cstar),\ell)\sum_{G' \in \mathcal{G}} \Pr{G'} \bmone_{(G',v)},
\]
which can be rewritten as
\begin{align*}
F(A_0) 
&= \sum_{r=1}^{m} \sum_{\ell=1}^r f(\ell)\Pr(\pi_v(\cstar),\ell)\sum_{G' \in \mathcal{G}} \Pr{G'} \sum_{v:\pi_v(\cstar)=r}\bmone_{(G',v)}
\\
&= \sum_{r=1}^{m} \sum_{\ell=1}^r f(\ell)\Pr(\pi_v(\cstar),\ell)\sum_{G' \in \mathcal{G}} \Pr{G'} |\{v_v\in R(G') \wedge \pi_v(\cstar)=r\}|
\\
&= \sum_{r=1}^{m} \sum_{\ell=1}^r f(\ell)\Pr(\pi_v(\cstar),\ell)\sum_{G' \in \mathcal{G}} \Pr{G'} |R_{G'}(A_0, V^r _{\cstar})|.
\end{align*}
Thus, $F(A_0)$ is a non-negative linear combination of the monotone submodular function $|R_{G'}(A_0, V^r _{\cstar})|$ (see Lemma~\ref{lem:submodularity}), and hence $F(A_0) - F(\emptyset)$ is also monotone and submodular. 
Thus, we can use \textsc{Greedy} (Algorithm~\ref{alg:greedy}) to find a $(1-1/e)$-approximation to the problem of maximizing the score of the target candidate~\cite{NWF78}.

\section{Approximating Margin of Victory}
\label{sec:approx}
We have seen in previous sections that we can map the problem of maximizing the score of the target candidate to that of influence maximization both in the \emph{plurality} (Section~\ref{sec:plurality}) and in the \emph{scoring} rules (Section~\ref{sec:scoring});
we also defined two alternative processes (Definitions~\ref{def:lcf} and~\ref{def:ldr}) and showed their equivalence to \LTM for both rules (Theorems~\ref{th:plurality} and~\ref{th:scoring}). 
By showing that the objective function is monotone and submodular w.r.t.\ the initial set of seed nodes (Lemma~\ref{lem:submodularity}) it follows that \textsc{Greedy} (Algorithm~\ref{alg:greedy}) finds a $(1-1/e)$-approximation of the optimum~\cite{NWF78}.

In the following we show how to achieve a constant factor approximation to the original problem of maximizing the \MOV by only maximizing the score of the target candidate.
Given the equivalence of the processes with \LTM, we can formulate our original objective function as the average $\MOV_{G'}$ computed on a sampled live-edge graph $G'$, namely $\Ex{\MOV(A_0)} = \Ex{\MOV_{G'}(A_0)}$, where
\[
\MOV_{G'}(A_0) = \mu_{G'}(\emptyset) - \mu_{G'}(A_0),
\]
and $\mu_{G'}$ is the change in margin on a fixed $G'$.
We formulate the margin on the live-edge graphs in a way that is akin to that of~\cite{wilder2018controlling}:
We can exploit such formulation to prove our constant factor approximation with the same proof structure since also in our case the objective function is monotone and submodular (Lemma~\ref{lem:submodularity}).
In particular in the simple case of the \emph{plurality} rule we have that
\begin{align*}
\Ex{\MOV_{G'}(A_0)} 
&:= \sum_{r=2}^{m} \frac{\alpha(r)}{r - 1} |R_{G'}(A_0, V^r _{\cstar})|
\\
&\quad + \min_{c_z}\left(
\max_{c_i}|V^1_{c_i}|
- |V_{c_z} ^1|
+ \sum_{r=2}^{m} \frac{\alpha(r)}{r - 1} |R_{G'}(A_0, V^r _{\cstar} \cap V^1 _{c_z})|
\right),
\end{align*}
where: 
the first term is the number of points gained by the target candidate after \LTM;
the second term (the first inside the minimum) is the number of points of the most voted opponent before \LTM;
the third is the total number of points that the most voted opponent after \LTM had before the process;
the fourth term is the number of points that the most voted opponent after \LTM lost because of the shifting of candidate $\cstar$.

Similarly, in the general case of arbitrary \emph{scoring} rules, we have
\begin{align*}
\Ex{\MOV_{G'}(A_0)} 
&:= \sum_{r=2}^{m}\sum_{\ell=1}^{r-1} \Pr(r, \ell) \, |R_{G'}(A_0, V^r _{\cstar})| \, (f(\ell) - f(r))
\\
&\qquad + \min_{c_z} \left(
\max_{c_i}\sum_{r=1}^{m}f(r)|V_{c_i} ^r|
- \sum_{r=1}^{m}f(r)|V_{c_z} ^r|
\right.
\\
&\qquad \left. + \sum_{r=2}^{m}\sum_{\ell=1}^{r-1}\sum_{h=\ell}^{r-1} \Pr(r, \ell) \, |R_{G'}(A_0, V^r _{\cstar}\cap V^h _{c_z})| \, (f(h) - f(h+1))
\right),
\end{align*}
where the meaning of the terms is similar to above.
This latter formulation is just a generalization of the plurality case whenever we choose $f$ such that $f(1) = 1$ and $f(r) = 0$, for each $r\in\{2,\ldots,m\}$.
In this way we would have that the gain in score would be just 1 and that $\frac{\alpha(r)}{r-1} = \Pr{r, 1}$.


In the following we prove that, up to the loss of a constant-factor in the approximation ratio, it suffices to concentrate only on the score of
the target candidate $\cstar$ and not on the margin w.r.t.\ the most voted opponent.
\begin{theorem}\label{th:approx}
\textsc{Greedy} (Algorithm~\ref{alg:greedy}) is a $\frac{1}{3}(1-1/e)$-approximation algorithm for the problem of election control in arbitrary scoring rule voting systems.
\end{theorem}
\begin{proof}
Let $A_0$ be the solution found by \textsc{Greedy} (Algorithm~\ref{alg:greedy}) in the election control problem and let $A_0^\star$ be the optimal solution.
Let $\bar{c}$ and $\hat{c}$ respectively be the candidates that minimize the second term of 
$\Ex{\MOV_{G'}(A_0)}$ and $\Ex{\MOV_{G'}(A_0^\star)}$.
Note that
\begin{align*}
\Ex{\MOV_{G'}(A_0)} 
&= F(A_0) - F(\emptyset) + |V^1_{c}| - |V^1 _{\bar{c}}|
\\
&\qquad + \sum_{r=2}^{m}\sum_{\ell=1}^{r-1}\sum_{h=\ell}^{r-1} \Pr(r, \ell) \, |R_{G'}(A_0, V^r _{\cstar}\cap V^h _{\bar{c}})| \, (f(h) - f(h+1)),
\end{align*}
where $c$ is the most voted candidate before the process.
Since $F(A_0)-F(\emptyset) = \sum_{r=2}^{m}\sum_{\ell=1}^{r-1} \Pr(r, \ell) \, |R_{G'}(A_0, V^r _{\cstar})| \, (f(\ell) - f(r))$
and $F(A_0) - F(\emptyset) \geq (1-1/e)(F(A_0^\star) - F(\emptyset))$, we get
\begin{align*}
&\Ex{\MOV_{G'}(A_0)} 
\geq F(A_0) - F(\emptyset) + |V^1_{c}| - |V^1 _{\bar{c}}|
\\
&\geq (1-1/e) \left[ \sum_{r=2}^{m}\sum_{\ell=1}^{r-1} \Pr(r, \ell) \, |R_{G'}(A_0^\star, V^r _{\cstar})| \, (f(\ell) - f(r))
\right] + |V^1_{c}| - |V^1 _{\bar{c}}|
\\
&\geq \frac{1}{3} (1-1/e) \left[
\sum_{r=2}^{m}\sum_{\ell=1}^{r-1} \Pr(r, \ell) \, |R_{G'}(A_0^\star, V^r _{\cstar})| \, (f(\ell) - f(r))
\right.
\\
&\qquad + \sum_{r=2}^{m}\sum_{\ell=1}^{r-1}\sum_{h=\ell}^{r-1} \Pr(r, \ell) \, |R_{G'}(A_0^\star, V^r _{\cstar}\cap V^h _{\bar{c}})| \, (f(h) - f(h+1))
\\
&\qquad \left. + \sum_{r=2}^{m}\sum_{\ell=1}^{r-1}\sum_{h=\ell}^{r-1} \Pr(r, \ell) \, |R_{G'}(A_0^\star, V^r _{\cstar}\cap V^h _{\hat{c}})| \, (f(h) - f(h+1)) + |V^1_{c}| - |V^1 _{\bar{c}}|
\right].
\end{align*}
Note that this is possible thanks to Theorem~\ref{th:scoring} and because the last two terms in the last inequality are smaller than the first term for any solution $A_0$ and candidate $c_i$ since the solution $A_0$ can only increase the score of $\cstar$. 
Therefore, for any other candidate $c_i$ the score can only decrease.
With some additional algebra we get that
\begin{align*}
    \Ex{\MOV_{G'}(A_0)} 
    &\geq \frac{1}{3} (1-1/e) \MOV_{G'}(A_0^*) + |V^1 _{\hat{c}}| - |V^1 _{\bar{c}}|
    \\
    &\qquad + \sum_{r=2}^{m}\sum_{\ell=1}^{r-1}\sum_{h=\ell}^{r-1} \Pr(r, \ell) \, |R_{G'}(A_0^\star, V^r _{\cstar}\cap V^h _{\bar{c}})| \, (f(h) - f(h+1)).
\end{align*}
By definition of $\hat{c}$ we have that 
\begin{align*}
|V^1 _{\hat{c}}| - |V^1 _{\bar{c}}| 
&\geq
\sum_{r=2}^{m}\sum_{\ell=1}^{r-1}\sum_{h=\ell}^{r-1} \Pr(r, \ell) \, |R_{G'}(A_0^\star, V^r _{\cstar}\cap V^h _{\hat{c}})| \, (f(h) - f(h+1)) 
\\
&\qquad- \sum_{r=2}^{m}\sum_{\ell=1}^{r-1}\sum_{h=\ell}^{r-1} \Pr(r, \ell) \, |R_{G'}(A_0^\star, V^r _{\cstar}\cap V^h _{\bar{c}})| \, (f(h) - f(h+1))
\end{align*}
and therefore
\[
\Ex{\MOV_{G'}(A_0)} \geq \frac{1}{3} (1-1/e) \MOV_{G'}(A_0^*).
\]
\end{proof}

\section{Destructive Election Control}
\label{sec:destructive}
In this section we focus on the \emph{destructive election control} problem.
The model is similar to the \emph{constructive} one (see Section~\ref{sec:problem}): Here we define, for each node $v \in V$, the number of positions of which $\cstar$ shifts down after the \LTM process as
\[
\pi^\downarrow_v(\cstar) := \min \left( m - \pi_v(\cstar), \, \Floor{ \frac{\alpha(\pi_v(\cstar))}{t_v} \sum_{\substack{u \in A,\,(u,v) \in E}} b_{uv} } \right).
\]
The final position of $\cstar$ in $v$ will be $\tilde{\pi}_v(\cstar) := \pi_v(\cstar) + \pi^\downarrow_v(\cstar)$ and the overall score that $\cstar$ gets is 
\[
F_D(A_0) := \Ex{\sum_{v\in V}f(\pi_v(\cstar) + \pi^\downarrow_v(\cstar) )}.
\]

Formally, the problem can be defined as that of finding an initial set of seed nodes $A_0$ such that
\[
\begin{array}{rl}
    \max_{A_0} & \Ex{\MOV_D(A_0)} := \Ex{\mu(A_0) - \mu(\emptyset)}
    \\
    \text{s.t.} & |A_0| \leq B,
\end{array}
\]
namely to find an initial set of seed nodes of at most size $B$ that maximizes the expected $\MOV_D$, i.e., minimizes the expected \MOV.


Similarly to the constructive case, we aim at decreasing the overall score of a target candidate $\cstar$ as much as possible since, as before, in this way we can achieve a constant factor approximation.
To do that we provide a reduction from the destructive to the constructive case. 
Given an instance of destructive control, we build an instance of constructive control in which we simply reverse the rankings of each node and complement the scoring function to its maximum value.
Roughly speaking, this reduction maintains invariant the absolute value of the change in margin of the score of any candidate between the two cases.
Formally, for each $v\in V$, the new instance has a preference list defined as $\pi'_v(c) := m-\pi_v(c)+1$ for each candidate $c \in C$, and, for each position $r\in\{1,\ldots,m\}$, has a scoring function defined as $f'(r) := f_{\max} - f(m-r+1)$, where $f_{\max} := \max_{r \in \{1,\ldots,m\}}{f(r)}$.
For each $v\in V$, the ranking of $\cstar$ in the new instance is $\pi'_v(\cstar) := m-\pi_v(\cstar)+1$.

For each solution $A_0$ found in the new instance, i.e., a constructive one, the overall score of $\cstar$ after the process is 
\[
    F'(A_0) := \Ex{\sum_{v\in V}f'(\pi'_v(\cstar) - \pi'^\uparrow_v(\cstar))},
\] 
where $\pi'^\uparrow_v(\cstar) := \min \left( \pi'_v(\cstar) -1, \, \Floor{ \frac{\alpha(\pi_v(\cstar))}{t_v} \sum_{\substack{u \in A,\,(u,v) \in E}} b_{uv} } \right)$.
Let $F_D(\emptyset) = F(\emptyset)$ and $F'(\emptyset) := \sum_{v\in V}f'(\pi'_v(\cstar))$. 
Then the following lemma holds.
\begin{lemma}\label{lem:reduction}
\(
F_D(\emptyset) -F_D(A_0) =  F'(A_0) -F'(\emptyset),
\)
for every $A_0$.
\end{lemma}
\begin{proof}
Observe that 
$\pi'^\uparrow_v(\cstar) = \pi^\downarrow_v(\cstar)$
and that $\pi_v(\cstar)=m-\pi'_v(\cstar)+1$. 
It follows that
\begin{align*}
F'(A_0) - F'(\emptyset) &= 
\Ex{\sum_{v\in V}[ f_{\max} - f(m - (\pi'_v(\cstar) - \pi'^\uparrow_v(\cstar))+1)]}
\\
&\qquad - \Ex{\sum_{v\in V}[ f_{\max} - f(m-\pi'_v(\cstar)+1 )]}
\\
&= \Ex{\sum_{v\in V}[ f(m-\pi'_v(\cstar)+1 ) - f(m - (\pi'_v(\cstar) - \pi'^\uparrow_v(\cstar))+1)]}
\\
&= \Ex{\sum_{v\in V}[ f(m-\pi'_v(\cstar)+1 ) - f(m - \pi'_v(\cstar) + 1 + \pi'^\uparrow_v(\cstar))]}
\\
&= \Ex{\sum_{v\in V}[ f(\pi_v(\cstar) ) - f(\pi_v(\cstar) + \pi_v^\downarrow(\cstar))]}
= F(\emptyset) - F_D(A_0).
\end{align*}
\end{proof}
The reduction, together with Lemma~\ref{lem:reduction}, allows us to maximize the score of the target candidate in the constructive case and then to map it back to destructive case.

\begin{theorem}
\textsc{Greedy} (Algorithm~\ref{alg:greedy}) is a $\frac{1}{2}(1-1/e)$-approximation algorithm for the problem of \emph{destructive} election control in arbitrary scoring rule voting systems.
\end{theorem}
\begin{proof}

Let $A_0$ be the solution found by \textsc{Greedy} (Algorithm~\ref{alg:greedy}) in the election control problem and let $A_0^\star$ be the optimal solution.
Let $\bar{c}$ and $\hat{c}$ respectively be the candidates that minimize the first term of $\Ex{\MOV_{D}(A_0)}$ and $\Ex{\MOV_{D}(A_0^\star)}$.
By Lemma~\ref{lem:reduction} we have that
\begin{align*}
    & \Ex{\MOV_D(A_0)} 
    = \Ex{\mu(A_0) - \mu(\emptyset)} 
    \\
    & = F(\emptyset) - F_D(A_0) - |V^1 _c| + |V^1 _{\bar{c}}|
    \\
    &\qquad + \sum_{r=1}^{m-1}\sum_{h=r+1}^{m}\sum_{\ell=r+1}^{m} \Pr(r, \ell) \, |R_{G'}(A_0, V^r _{\cstar}\cap V^h _{\bar{c}})| \, (f(h-1) - f(h))
    \\
    & = F'(A_0) - F'(\emptyset) - |V^1 _c| + |V^1 _{\bar{c}}| 
    \\
    &\qquad + \sum_{r=1}^{m-1}\sum_{h=r+1}^{m}\sum_{\ell=r+1}^{m} \Pr(r, \ell) \, |R_{G'}(A_0, V^r _{\cstar}\cap V^h _{\bar{c}})| \, (f(h-1) - f(h))
\end{align*}
where $c$ is the most voted candidate before the process.
Since $F'(A_0) - F'(\emptyset)$ is an instance of the score in the constructive case we able to approximate this value, thus we get
\begin{align*}
    &\Ex{\MOV_D(A_0)} \geq \left(1 - \frac{1}{e}\right) \left[ F'(A_0^*) - F'(\emptyset) - |V^1 _c| + |V^1 _{\bar{c}}| \right.
    \\
    &\qquad\left.
    + \sum_{r=1}^{m-1}\sum_{h=r+1}^{m}\sum_{\ell=r+1}^{m} \Pr(r, \ell) \, |R_{G'}(A_0, V^r _{\cstar}\cap V^h _{\bar{c}})| \, (f(h-1) - f(h)) \right]
    \\
    &\geq \frac{1}{2}\left(1 - \frac{1}{e}\right) \left[ F(\emptyset) - F_D(A_0^\star) - |V^1 _c| + |V^1 _{\bar{c}}|
    \right.
    \\
    &\qquad\left.
    +\sum_{r=1}^{m-1}\sum_{h=r+1}^{m}\sum_{\ell=r+1}^{m} \Pr(r, \ell) \, |R_{G'}(A_0^\star, V^r _{\cstar}\cap V^h _{\hat{c}})| \, (f(h-1) - f(h))\right.
    \\
    &\qquad\left.
    +\sum_{r=1}^{m-1}\sum_{h=r+1}^{m}\sum_{\ell=r+1}^{m} \Pr(r, \ell) \, |R_{G'}(A_0, V^r _{\cstar}\cap V^h _{\bar{c}})| \, (f(h-1) - f(h)) + |V^1 _{\hat{c}}| - |V^1 _{\hat{c}}| \right]
    \\
    &\geq \frac{1}{2}\left(1 - \frac{1}{e}\right) \left[ \MOV_D(A^\star) + |V^1 _{\bar{c}}|
    \right.
    \\
    &\qquad\left.
    +\sum_{r=1}^{m-1}\sum_{h=r+1}^{m}\sum_{\ell=r+1}^{m} \Pr(r, \ell) \, |R_{G'}(A_0^, V^r _{\cstar}\cap V^h _{\bar{c}})| \, (f(h-1) - f(h)) - |V^1 _{\hat{c}}| \right]
\end{align*}
By definition of $\bar{c}$ we have that 
\begin{align*}
|V^1 _{\bar{c}}| - |V^1 _{\hat{c}}| 
&\geq \sum_{r=1}^{m-1}\sum_{h=r+1}^{m}\sum_{\ell=r+1}^{m} \Pr(r, \ell) \, |R_{G'}(A_0, V^r _{\cstar}\cap V^h _{\hat{c}})| \, (f(h-1) - f(h))
\\
&\qquad- \sum_{r=1}^{m-1}\sum_{h=r+1}^{m}\sum_{\ell=r+1}^{m} \Pr(r, \ell) \, |R_{G'}(A_0, V^r _{\cstar}\cap V^h _{\bar{c}})| \, (f(h-1) - f(h))
\end{align*}
and therefore
\[
\MOV_D(A_0) \geq \frac{1}{2} \left(1 - \frac{1}{e}\right) \MOV_D(A_0^*).
\]
\end{proof}

\section{Simulations}
In this section we present some experimental results that show how our approximation algorithm performs on real-world networks.
We chose four heterogeneous social and communication networks on which political campaigning messages could spread, namely:
\begin{enumerate}
    \item \emph{facebook},\footnote{\label{fn:konect}\url{http://konect.uni-koblenz.de}} an undirected network of 10 Facebook users, with 2,888 nodes and 2,981 edges;
    \item \emph{irvine},\fnref{\ref{fn:konect}} a directed network of instant messages exchanged between students at U.C.\ Irvine, with 1,899 nodes and 20,296 edges;
    \item \emph{netscience},\footnote{\label{fn:netdata}\url{http://www.umich.edu/~mejn/netdata/}} an undirected network of research collaborations in network science, with 1,461 nodes and 2,742 edges;
    \item \emph{polblogs},\fnref{\ref{fn:netdata}} a directed network of hyperlinks between web blogs on US politics, with 1,224 nodes and 19,025 edges.
\end{enumerate}
We made the two undirected networks directed, by doubling the edges and orienting them;
moreover, to adhere to the Linear Threshold Model, we assigned random weights to edges of the graphs since they are unweighted.
Recent literature in influence maximization shows that advanced techniques can be used to scale our approximation algorithm to much larger networks~\cite{tang2014influence}.

\smallskip
We considered three different scenarios, each with a different number of candidates, i.e., $m=2,5,10$.
For each scenario, we assigned a random preference list to each node of the networks; this assignment was performed 10 distinct times, by randomly permuting its preference list.
We separately analyzed three different initial budgets, i.e., $B=5,10,15$,
and three different values of $\alpha$ (the rate at which the position of candidate $c^\star$ changes in the preference list of each node), i.e., $\alpha=0.1,0.5,1$.
For each combination of parameters (dataset, number of candidates, preference list assignment, budget, and $\alpha$) we performed 20 experiments for each of the two considered voting systems,
namely \emph{plurality rule} and \emph{borda count} (as example for the \emph{scoring rule}), in the \emph{constructive election control} scenario.
We measured the \emph{Probability of Victory} (\POV), i.e., the fraction of times $c^\star$ won out of the 20 experiments, and the \emph{Margin of Victory} (\MOV), average value of the difference between the score of candidate $c^\star$ and the score of the most voted opponent.%
\footnote{Note that in this section we consider $\MOV = -\mu(A_0)$, i.e., we do not consider the initial difference $\mu(\emptyset)$, allowing negative values (we are not computing approximation ratios).}

\smallskip
All the experiments were run in parallel on a machine with four 16-core AMD Opteron\texttrademark{} 6376 with 2.3 GHz CPU, 16 MB L2 cache, 64 GB RAM, running Ubuntu 16.04.5 LTS. 
Overall, we performed 43,200 distinct runs for an average running time of approximately 15 minutes per run.

\begin{figure}[p!]
\centering
\begin{subfigure}{0.33\textwidth}            
    \includegraphics[width=\textwidth]{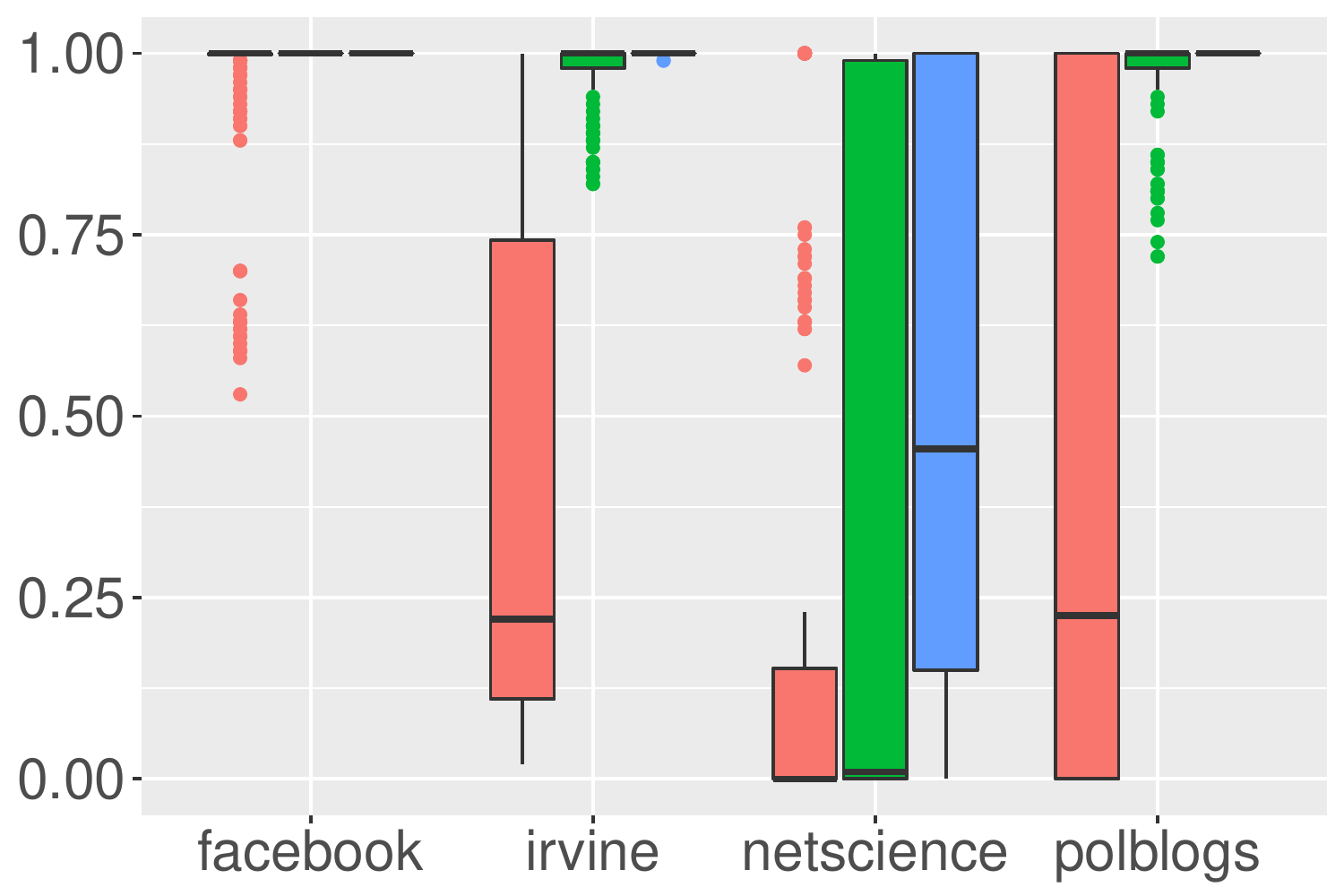}
    \caption{\POV, \emph{plurality rule}, B=5}
    \label{fig:pov_p5}
\end{subfigure}\hfill%
\begin{subfigure}{0.33\textwidth}            
    \includegraphics[width=\textwidth]{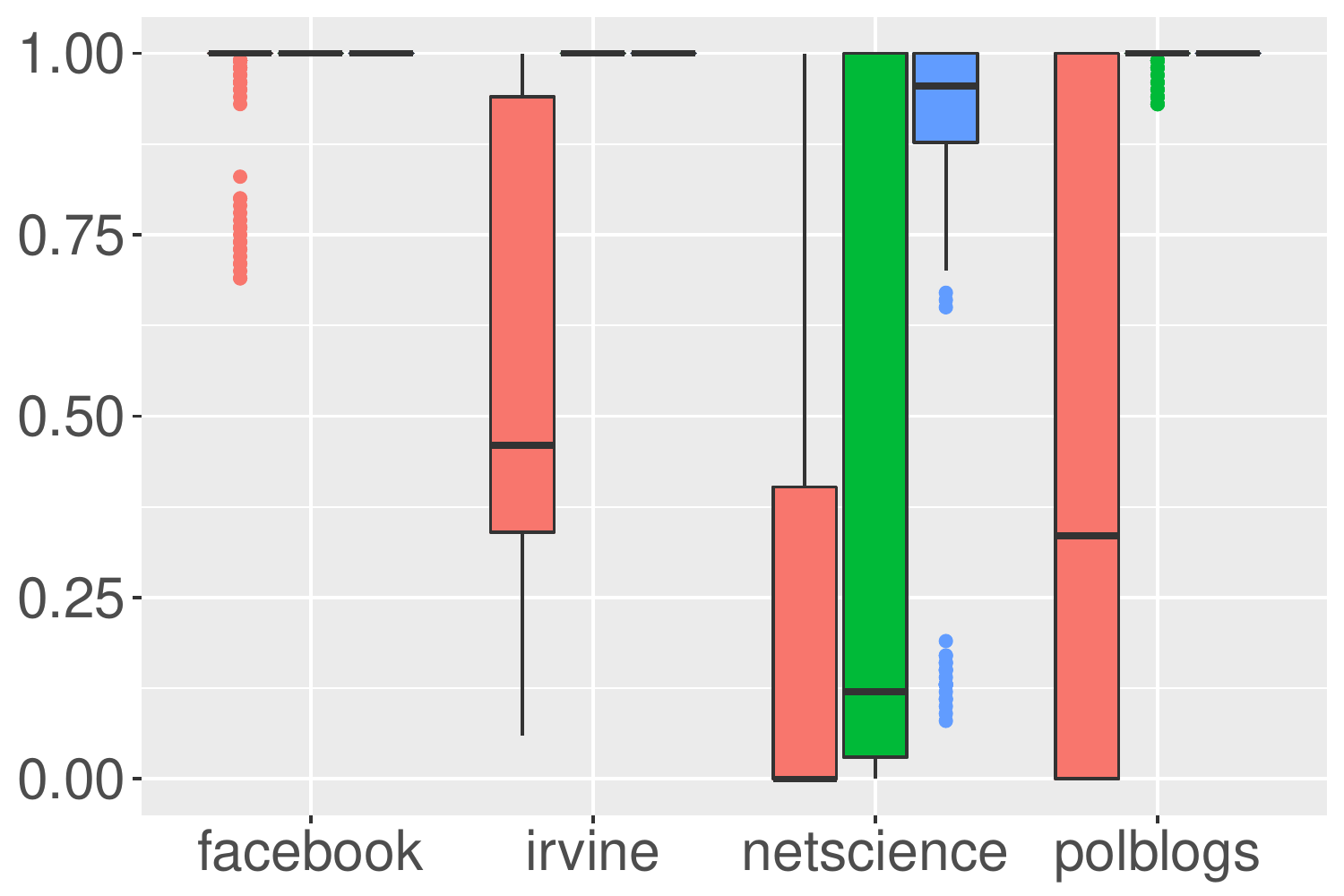}
    \caption{\POV, \emph{plurality rule}, B=10}
    \label{fig:pov_p10}
\end{subfigure}\hfill%
\begin{subfigure}{0.33\textwidth}            
    \includegraphics[width=\textwidth]{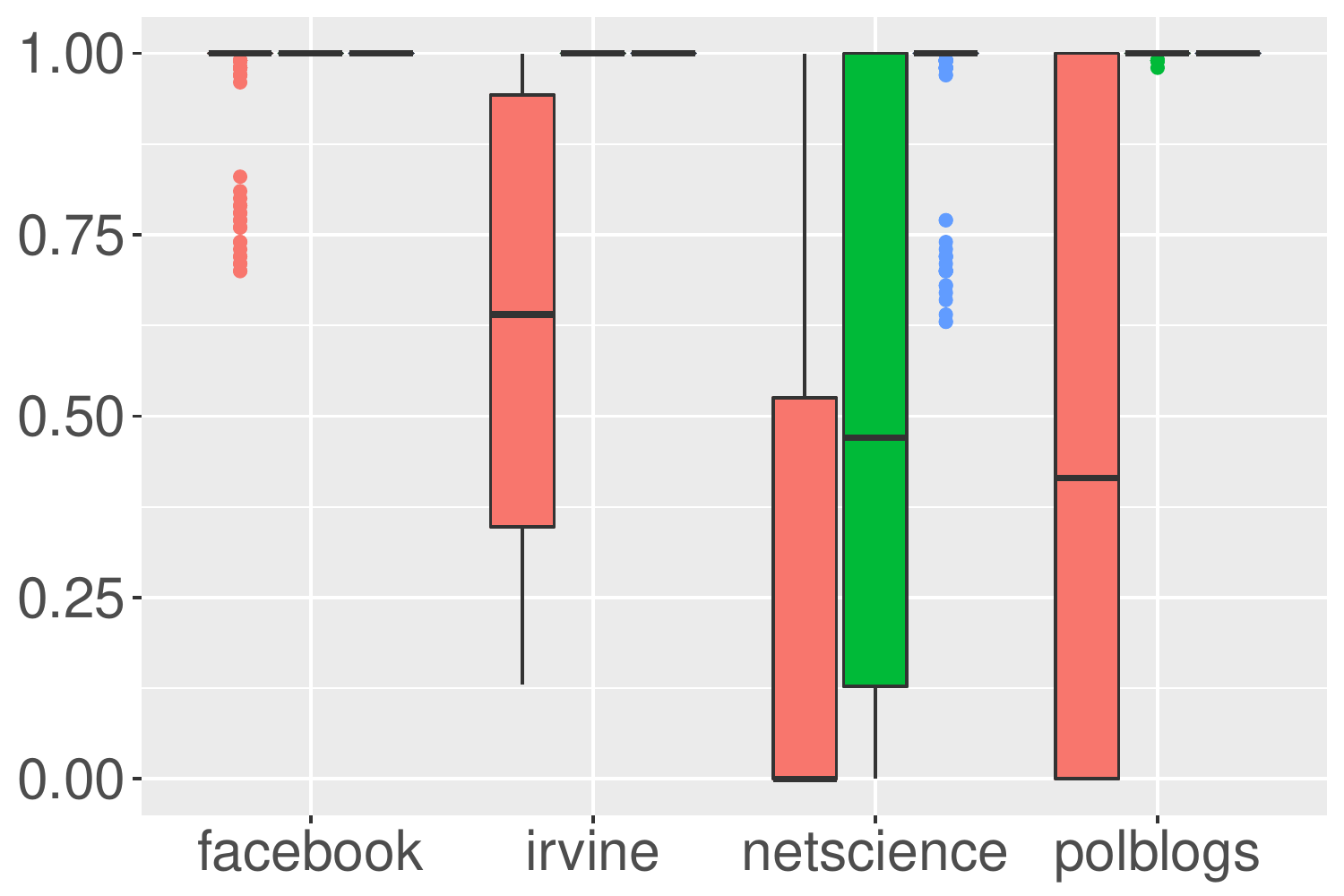}
    \caption{\POV, \emph{plurality rule}, B=15}
    \label{fig:pov_p15}
\end{subfigure}
\begin{subfigure}{0.33\textwidth}            
    \includegraphics[width=\textwidth]{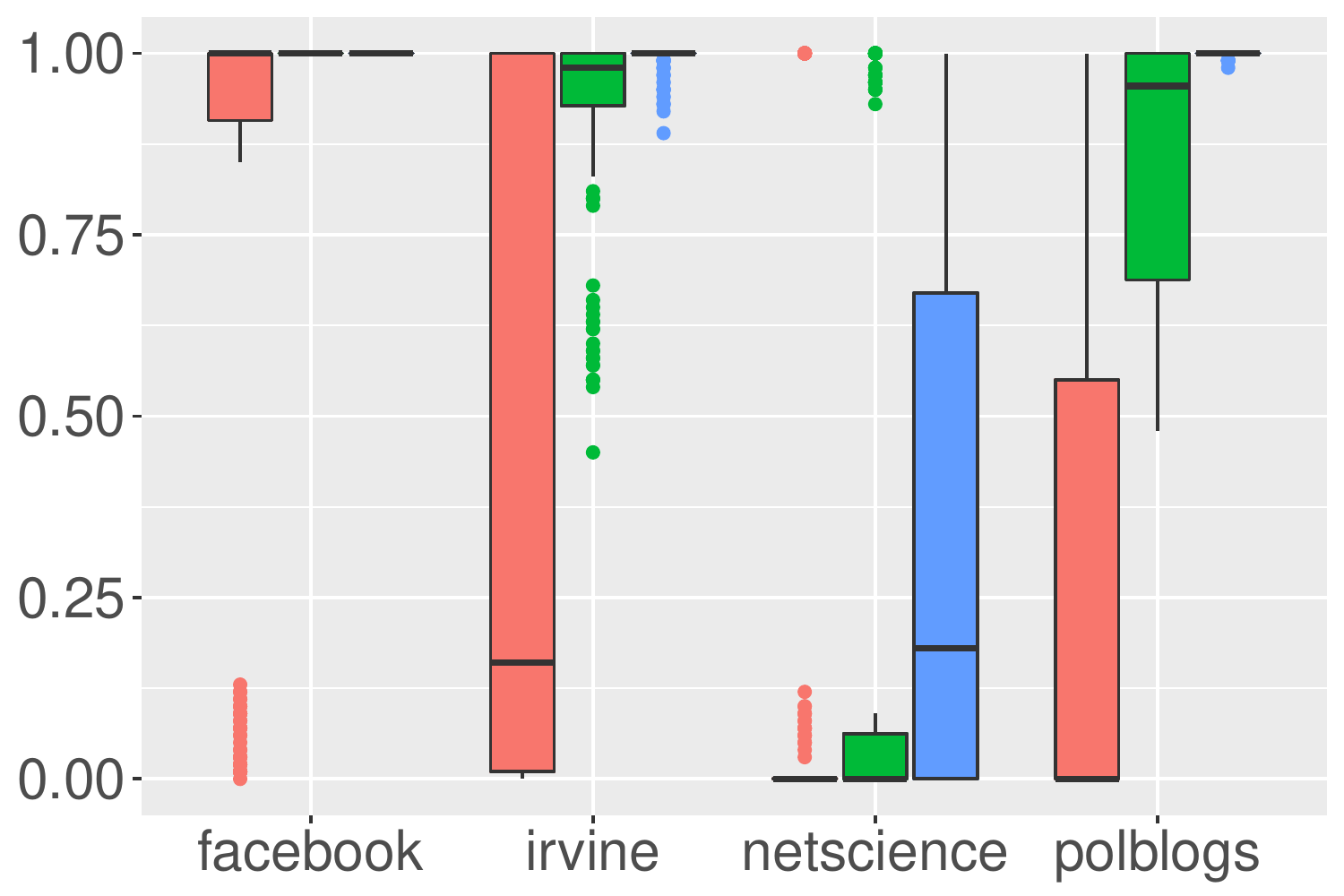}
    \caption{\POV, \emph{borda count}, B=5}
    \label{fig:pov_b5}
\end{subfigure}\hfill%
\begin{subfigure}{0.33\textwidth}            
    \includegraphics[width=\textwidth]{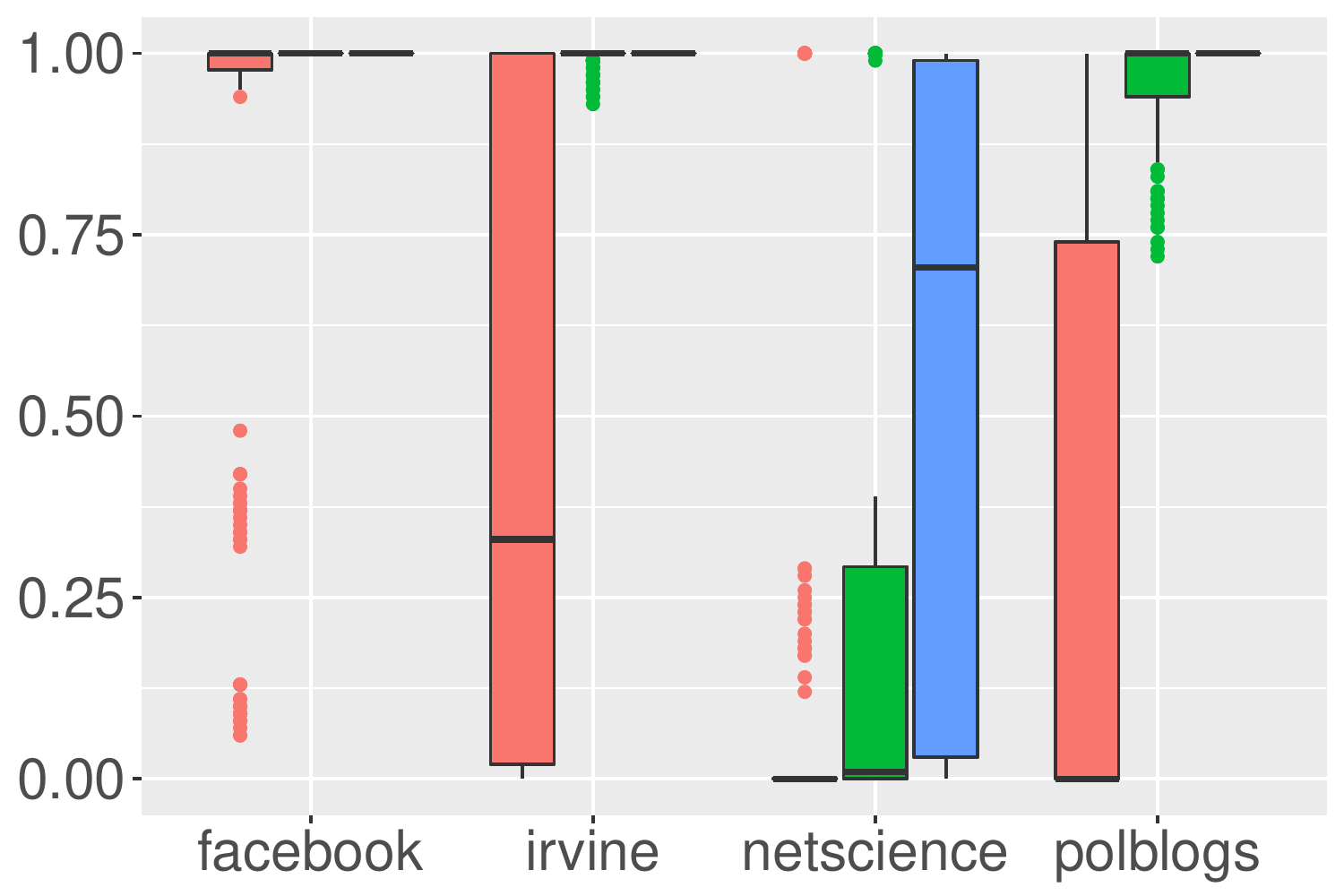}
    \caption{\POV, \emph{borda count}, B=10}
    \label{fig:pov_b10}
\end{subfigure}\hfill%
\begin{subfigure}{0.33\textwidth}            
    \includegraphics[width=\textwidth]{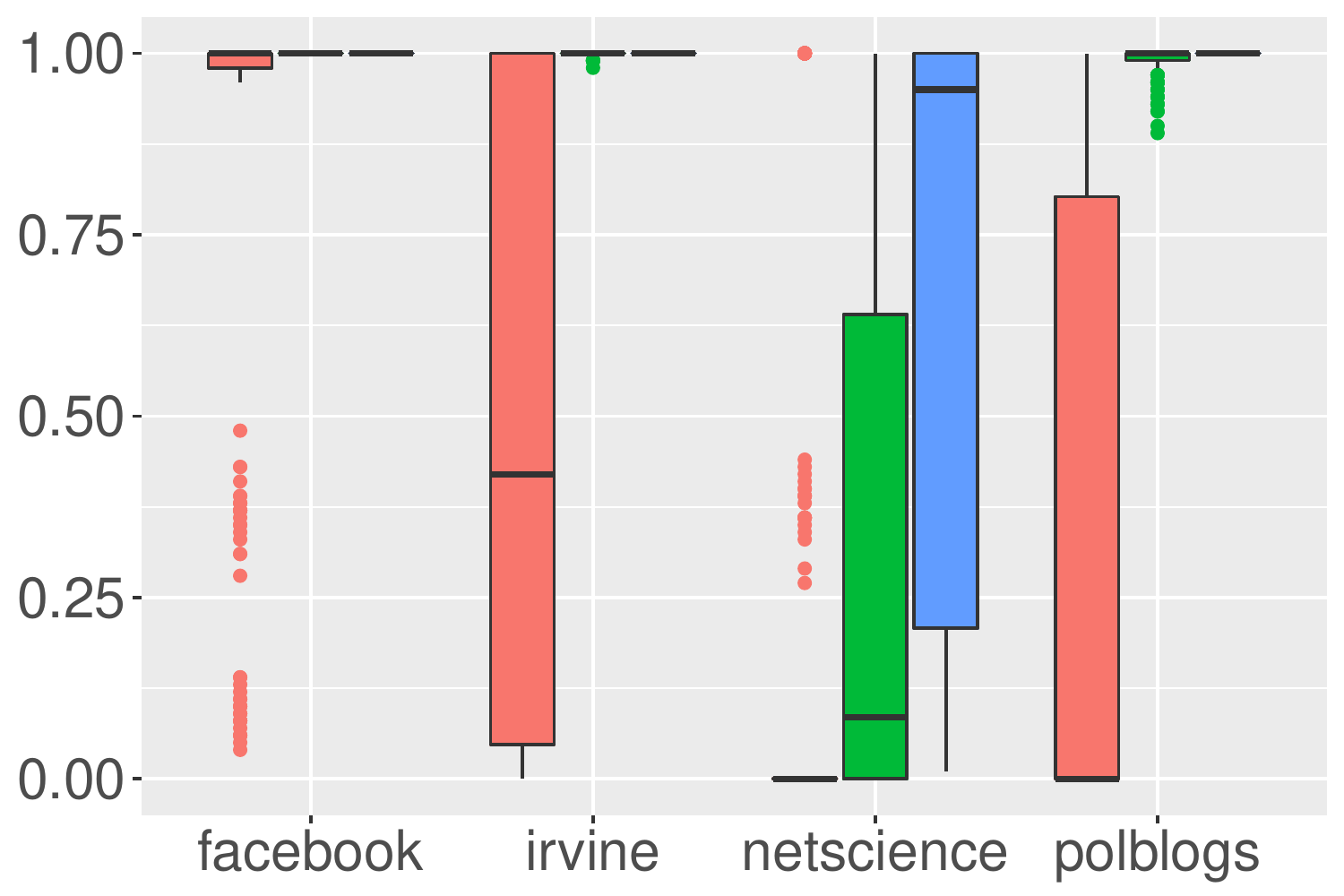}
    \caption{\POV, \emph{borda count}, B=15}
    \label{fig:pov_b15}
\end{subfigure}
\begin{subfigure}{0.33\textwidth}            
    \includegraphics[width=\textwidth]{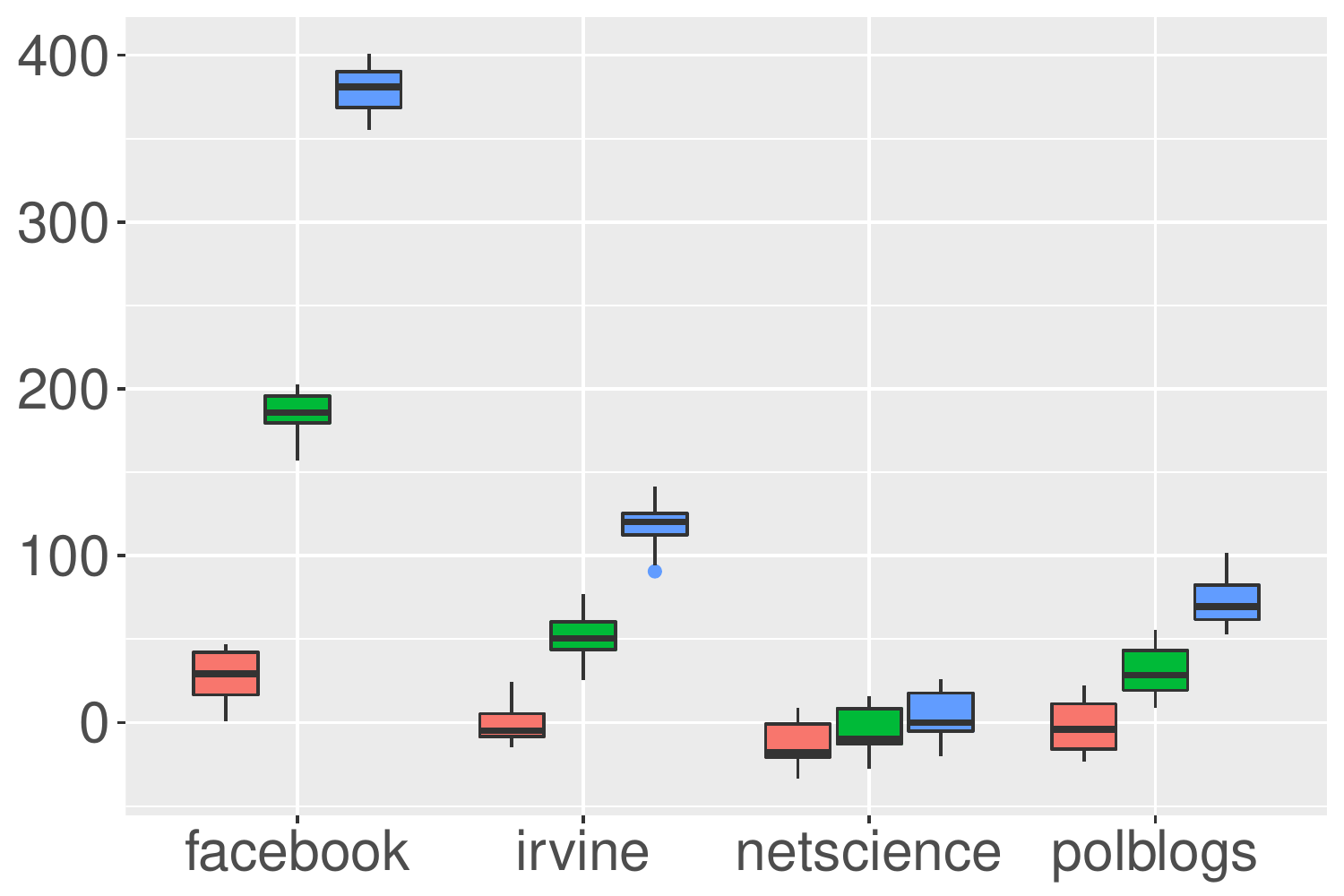}
    \caption{\MOV, \emph{plurality rule}, B=5}
    \label{fig:mov_p5}
\end{subfigure}\hfill%
\begin{subfigure}{0.33\textwidth}            
    \includegraphics[width=\textwidth]{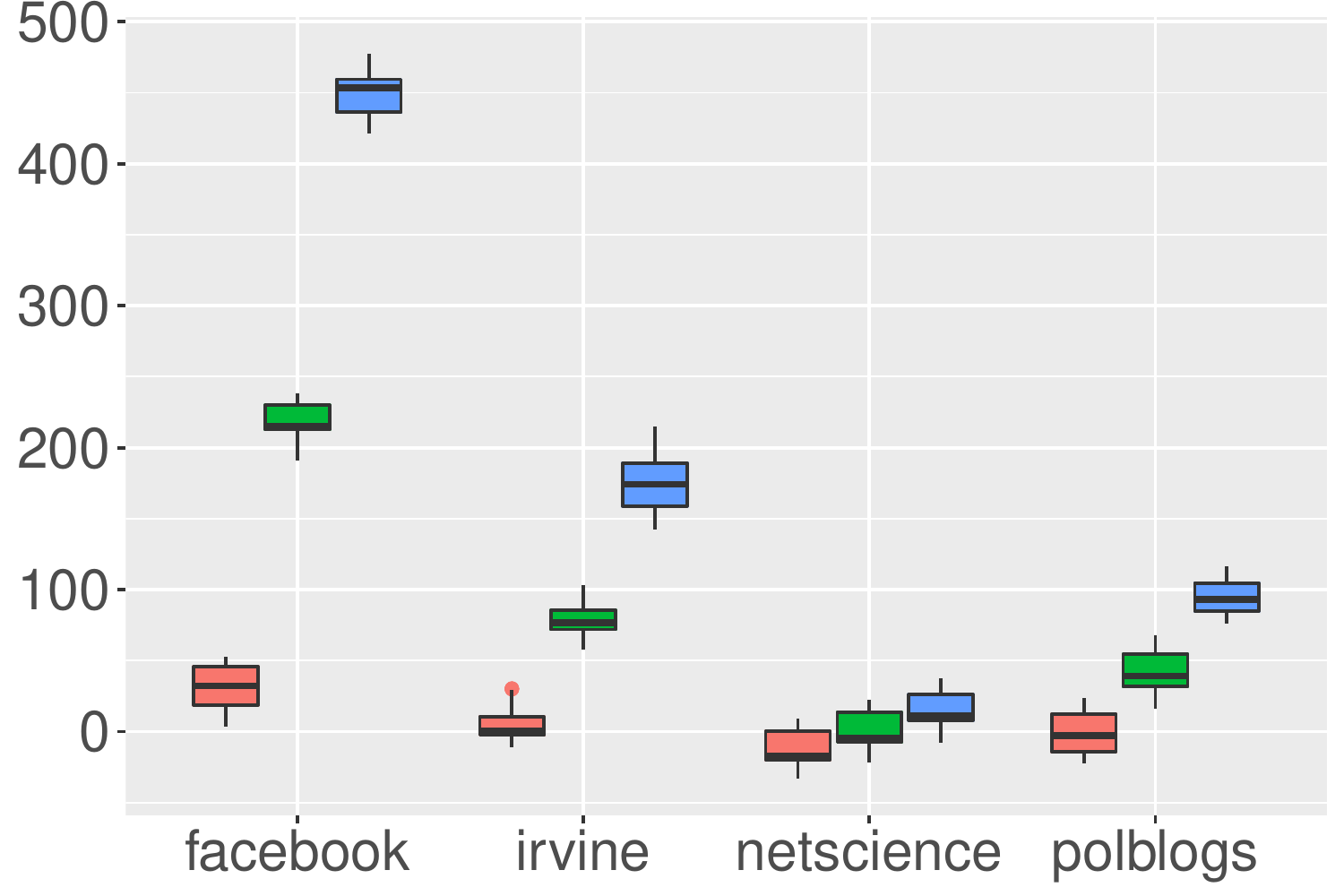}
    \caption{\MOV, \emph{plurality rule}, B=10}
    \label{fig:mov_p10}
\end{subfigure}\hfill%
\begin{subfigure}{0.33\textwidth}            
    \includegraphics[width=\textwidth]{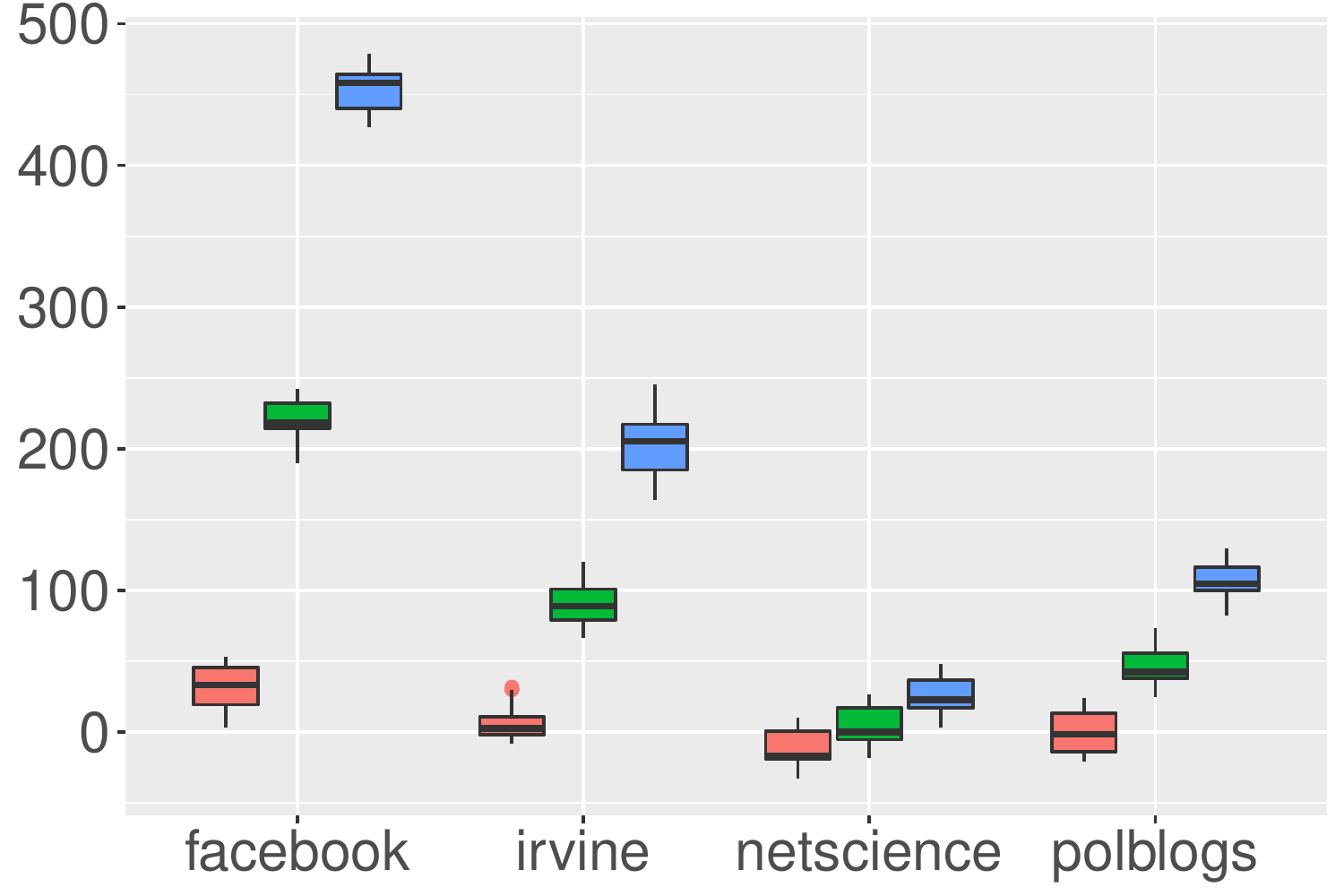}
    \caption{\MOV, \emph{plurality rule}, B=15}
    \label{fig:mov_p15}
\end{subfigure}
\begin{subfigure}{0.33\textwidth}            
    \includegraphics[width=\textwidth]{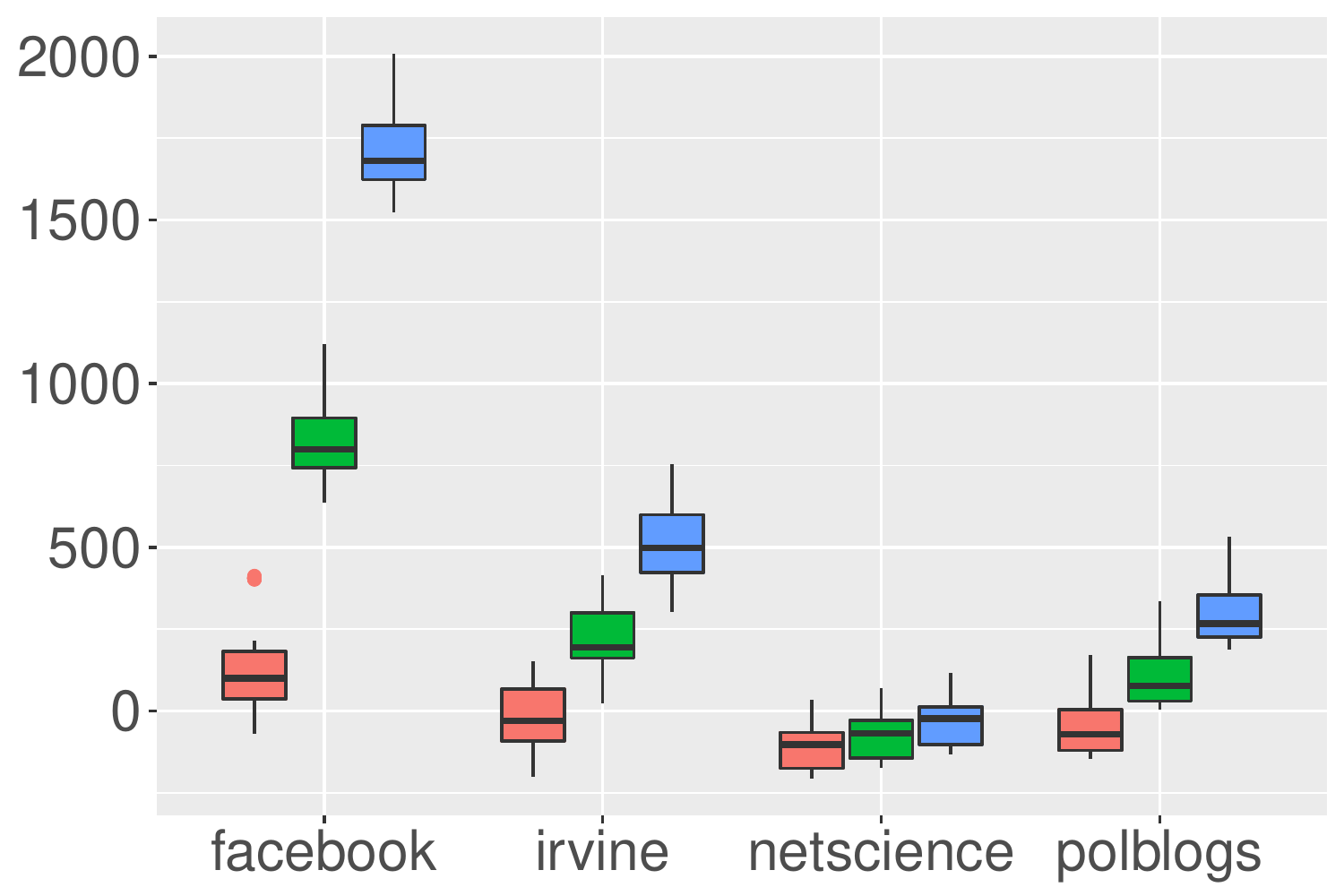}
    \caption{\MOV, \emph{borda count}, B=5}
    \label{fig:mov_b5}
\end{subfigure}\hfill%
\begin{subfigure}{0.33\textwidth}            
    \includegraphics[width=\textwidth]{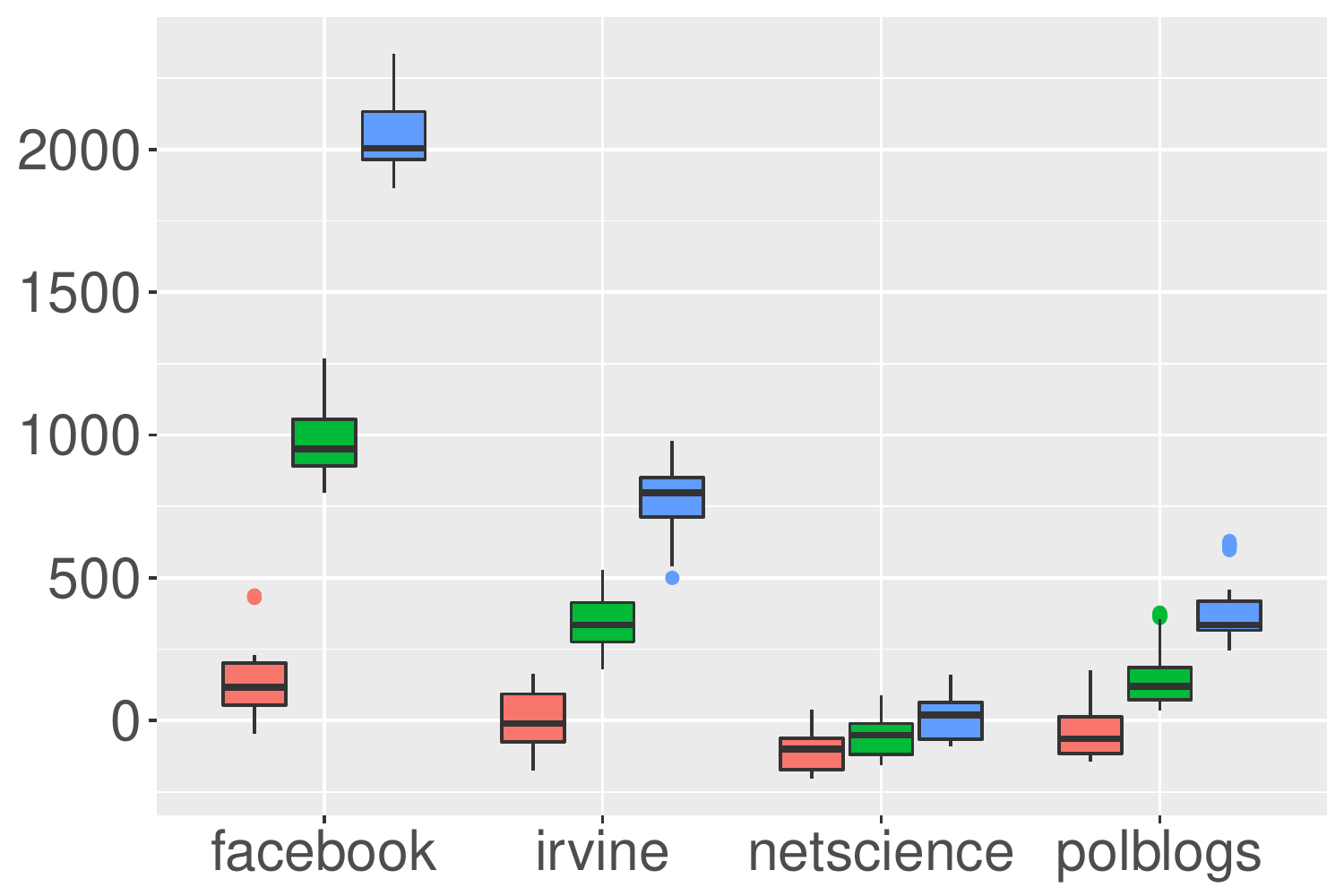}
    \caption{\MOV, \emph{borda count}, B=10}
    \label{fig:mov_b10}
\end{subfigure}\hfill%
\begin{subfigure}{0.33\textwidth}            
    \includegraphics[width=\textwidth]{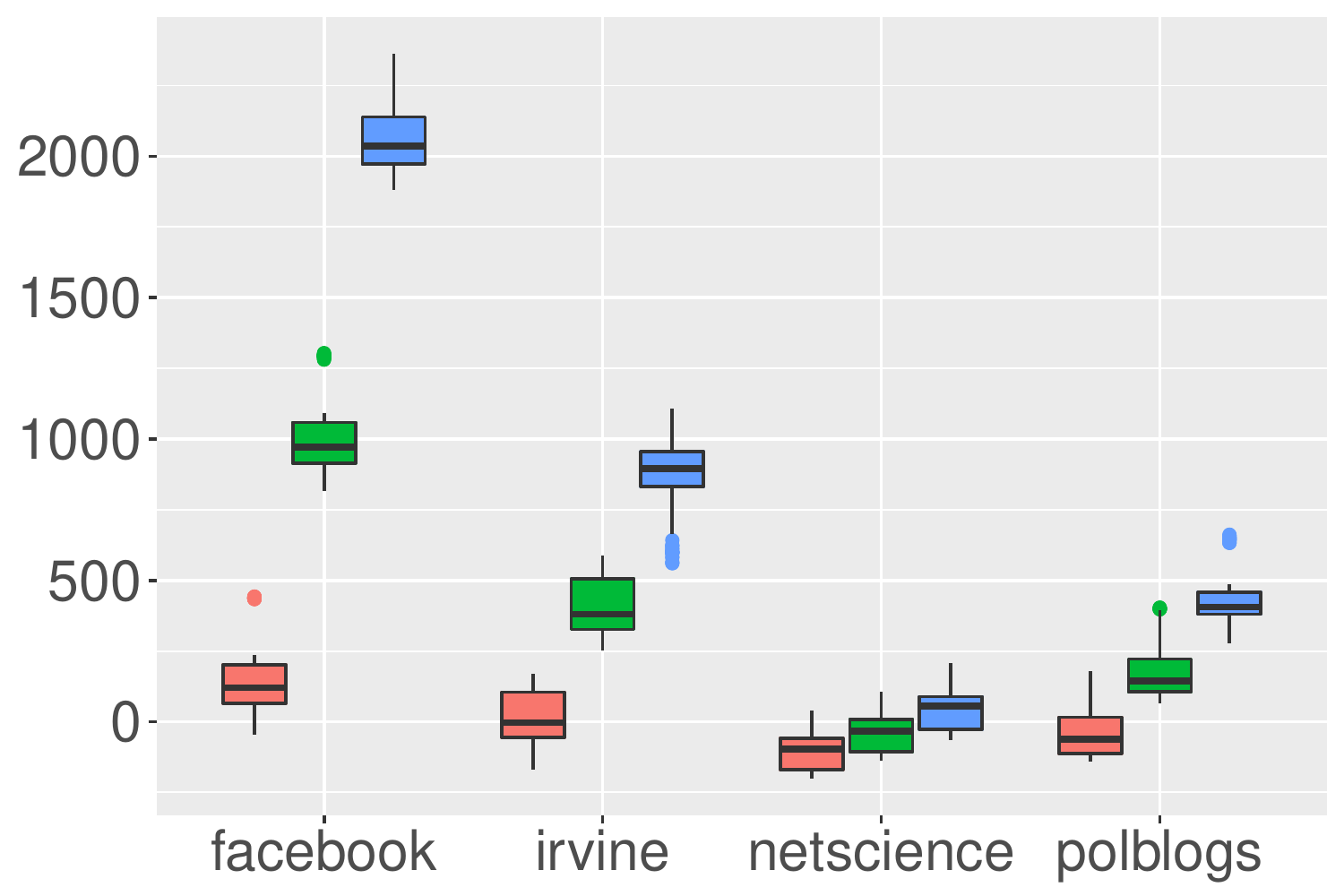}
    \caption{\MOV, \emph{borda count}, B=15}
    \label{fig:mov_b15}
\end{subfigure}
\caption{\POV and \MOV with $m=10$. Each plot compares the results on different datasets (\emph{facebook}, \emph{irvine}, \emph{netscience}, \emph{polblogs}) and for different values of $\alpha$. For each dataset, from left to right: $\alpha=0.1$ (\emph{red}), $\alpha=0.5$ (\emph{green}), $\alpha=1.0$ (\emph{blue}).}
\label{fig:experiments}
\end{figure}

Figure~\ref{fig:experiments} shows the effectiveness of the algorithm in the scenario with $m=10$ candidates running for the elections and a fixed budget $B=5, 10, 15$ using as voting system the plurality rule and borda count.

We study the impact of $\alpha$ on the behavior of our algorithm (Figure~\ref{fig:experiments}).
The algorithm succeeds for $\alpha=0.5,1.0$, making candidate $c^\star$ win all or most of the times on all datasets but \emph{netscience}: We believe this is due to the topology of the dataset, with 267 connected components, that limitates the influence diffusion process.
Instead, when $\alpha=0.1$ it is difficult for the initially targeted voters to influence their friends in all datasets but \emph{facebook}, where most of the nodes have few friends and can be easily influenced by the 10 users on which the network is centered.
In fact the lower is the value of $\alpha$ the higher is the rate at which $c^\star$ shifts in the preference lists of the voters.

\smallskip
Tables~\ref{tab:experiments_2},~\ref{tab:experiments_5}, and~\ref{tab:experiments_10} report detailed results of the experiments.
Table~\ref{tab:experiments_2} reports unified results for \emph{plurality rule} and \emph{borda count} voting systems, given their equivalence in the scenario with only $m=2$ candidates running for the elections.
Figure~\ref{fig:experiments} gives a visual interpretation of the results, considering the scenario with $m=10$, which is the ``hardest'' among the considered ones, since it has the maximum number of candidates and the minimum budget.
Each boxplot considers 200 observations, i.e., the results obtained by permuting the preference list of each voter 10 times and repeating 20 experiments on each of them, and shows the results for all considered values of $B$.

\begin{landscape}
\begin{table}[pt!]
\centering
\caption{\POV and \MOV values relative to the experiments with $m=2$.}
\label{tab:experiments_2}
\begin{tabular}{crrrrrrrrrrrrr}
\toprule
& & \multicolumn{6}{c}{\POV} & \multicolumn{6}{c}{\MOV} \\
\cmidrule(lr){3-8}
\cmidrule(lr){9-14}
& & \multicolumn{2}{c}{$B=5$} & \multicolumn{2}{c}{$B=10$} & \multicolumn{2}{c}{$B=15$} 
& \multicolumn{2}{c}{$B=5$} & \multicolumn{2}{c}{$B=10$} & \multicolumn{2}{c}{$B=15$} \\ 
& \multicolumn{1}{c}{$\alpha$} & \multicolumn{1}{c}{$\mu$} & \multicolumn{1}{c}{$\sigma$} & \multicolumn{1}{c}{$\mu$} & \multicolumn{1}{c}{$\sigma$}  & \multicolumn{1}{c}{$\mu$} & \multicolumn{1}{c}{$\sigma$} & \multicolumn{1}{c}{$\mu$} & \multicolumn{1}{c}{$\sigma$} & \multicolumn{1}{c}{$\mu$} & \multicolumn{1}{c}{$\sigma$} & \multicolumn{1}{c}{$\mu$} & \multicolumn{1}{c}{$\sigma$} \\ 
\midrule
\multirow{3}{*}{\rotatebox[origin=c]{90}{\emph{facebook}}}
& 0.1
& 1.00 & 0.00 & 1.00 & 0.00 & 1.00 & 0.00 & 149.78 & 46.10 & 165.63 & 47.32 & 168.22 & 46.14 \\
& 0.5
& 1.00 & 0.00 & 1.00 & 0.00 & 1.00 & 0.00 & 631.90 & 39.11 & 743.61 & 35.70 & 751.42 & 38.55 \\
& 1.0
& 1.00 & 0.00 & 1.00 & 0.00 & 1.00 & 0.00 & 1234.91 & 31.05 & 1472.39 & 28.42 & 1476.18 & 28.37 \\
\midrule
\multirow{3}{*}{\rotatebox[origin=c]{90}{\emph{irvine}}}
& 0.1
& 0.95 & 0.15 & 0.96 & 0.13 & 0.99 & 0.03 & 66.93 & 39.44 & 82.04 & 40.21 & 89.84 & 38.75 \\
& 0.5
& 1.00 & 0.00 & 1.00 & 0.00 & 1.00 & 0.00 & 230.48 & 36.68 & 313.38 & 23.76 & 350.47 & 19.72 \\
& 1.0
& 1.00 & 0.00 & 1.00 & 0.00 & 1.00 & 0.00 & 441.67 & 40.74 & 624.90 & 39.16 & 693.16 & 59.91 \\
\midrule
\multirow{3}{*}{\rotatebox[origin=c]{90}{\emph{netscience}}}
& 0.1
& 0.74 & 0.42 & 0.78 & 0.39 & 0.81 & 0.36 & 29.95 & 40.73 & 32.55 & 40.80 & 34.71 & 40.86 \\
& 0.5
& 0.99 & 0.04 & 1.00 & 0.00 & 1.00 & 0.00 & 55.38 & 39.76 & 71.31 & 40.03 & 86.40 & 40.77 \\
& 1.0
& 1.00 & 0.00 & 1.00 & 0.00 & 1.00 & 0.00 & 85.98 & 39.57 & 122.36 & 40.12 & 151.88 & 41.31 \\
\midrule
\multirow{3}{*}{\rotatebox[origin=c]{90}{\emph{polblogs}}}
& 0.1
& 0.91 & 0.29 & 0.92 & 0.25 & 0.93 & 0.22 & 47.16 & 30.39 & 52.69 & 30.57 & 55.24 & 30.40 \\
& 0.5
& 1.00 & 0.00 & 1.00 & 0.00 & 1.00 & 0.00 & 150.11 & 29.34 & 178.11 & 28.84 & 196.28 & 27.95 \\
& 1.0
& 1.00 & 0.00 & 1.00 & 0.00 & 1.00 & 0.00 & 280.13 & 27.29 & 344.81 & 23.40 & 387.12 & 17.19 \\
\bottomrule
\end{tabular}

\footnotesize{$\mu$ and $\sigma$ are, respectively, the mean and the standard deviation of the observations averaged over the $10$ preference list permutations.}
\end{table}
\end{landscape}
\begin{landscape}
\begin{table*}[p!]
\centering
\caption{\POV and \MOV values relative to the experiments with $m=5$.}
\label{tab:experiments_5}
\begin{tabular}{clrrrrrrrrrrrrr}
\toprule
& & & \multicolumn{6}{c}{\POV} & \multicolumn{6}{c}{\MOV} \\
\cmidrule(lr){4-9}
\cmidrule(lr){10-15}
& & & \multicolumn{2}{c}{$B=5$} & \multicolumn{2}{c}{$B=10$} & \multicolumn{2}{c}{$B=15$} 
& \multicolumn{2}{c}{$B=5$} & \multicolumn{2}{c}{$B=10$} & \multicolumn{2}{c}{$B=15$} \\ 
& & \multicolumn{1}{c}{$\alpha$} & \multicolumn{1}{c}{$\mu$} & \multicolumn{1}{c}{$\sigma$} & \multicolumn{1}{c}{$\mu$} & \multicolumn{1}{c}{$\sigma$}  & \multicolumn{1}{c}{$\mu$} & \multicolumn{1}{c}{$\sigma$} & \multicolumn{1}{c}{$\mu$} & \multicolumn{1}{c}{$\sigma$} & \multicolumn{1}{c}{$\mu$} & \multicolumn{1}{c}{$\sigma$} & \multicolumn{1}{c}{$\mu$} & \multicolumn{1}{c}{$\sigma$} \\ 
\midrule
\multirow{6}{*}{\rotatebox[origin=c]{90}{\emph{facebook}}} &
\multirow{3}{*}{\rotatebox[origin=c]{90}{Plurality}}
& 0.1
& 0.86 & 0.29 & 0.95 & 0.12 & 0.95 & 0.10 & 47.52 & 30.88 & 55.56 & 29.91 & 57.08 & 29.96 \\
& & 0.5
& 1.00 & 0.00 & 1.00 & 0.00 & 1.00 & 0.00 & 302.92 & 30.50 & 362.70 & 28.65 & 366.13 & 28.50 \\
& & 1.0
& 1.00 & 0.00 & 1.00 & 0.00 & 1.00 & 0.00 & 621.27 & 29.22 & 745.84 & 30.36 & 749.99 & 28.11 \\
\cline{2-15}
& \multirow{3}{*}{\rotatebox[origin=c]{90}{Borda}} 
& 0.1
& 0.80 & 0.42 & 0.80 & 0.42 & 0.80 & 0.42 & 80.93 & 114.03 & 96.43 & 114.99 & 99.04 & 114.68 \\
& & 0.5
& 1.00 & 0.00 & 1.00 & 0.00 & 1.00 & 0.00 & 629.83 & 109.08 & 751.63 & 106.74 & 757.55 & 106.06 \\
& & 1.0
& 1.00 & 0.00 & 1.00 & 0.00 & 1.00 & 0.00 & 1314.98 & 103.33 & 1572.05 & 98.62 & 1588.18 & 99.25 \\
\midrule
\multirow{6}{*}{\rotatebox[origin=c]{90}{\emph{irvine}}} &
\multirow{3}{*}{\rotatebox[origin=c]{90}{Plurality}}
& 0.1
& 0.68 & 0.40 & 0.77 & 0.36 & 0.81 & 0.34 & 10.96 & 22.64 & 18.46 & 22.42 & 22.53 & 22.34 \\
& & 0.5
& 0.99 & 0.03 & 1.00 & 0.00 & 1.00 & 0.00 & 98.33 & 22.87 & 143.74 & 16.19 & 171.63 & 15.92 \\
& & 1.0
& 1.00 & 0.00 & 1.00 & 0.00 & 1.00 & 0.00 & 211.20 & 25.12 & 296.33 & 25.73 & 341.69 & 31.26 \\
\cline{2-15}
& \multirow{3}{*}{\rotatebox[origin=c]{90}{Borda}} 
& 0.1
& 0.61 & 0.41 & 0.71 & 0.41 & 0.74 & 0.41 & 15.32 & 76.67 & 29.68 & 74.70 & 37.03 & 75.88 \\
& & 0.5
& 0.96 & 0.12 & 1.00 & 0.01 & 1.00 & 0.00 & 203.44 & 74.23 & 307.97 & 70.30 & 359.88 & 64.62 \\
& & 1.0
& 1.00 & 0.01 & 1.00 & 0.00 & 1.00 & 0.00 & 443.95 & 77.76 & 624.60 & 76.87 & 716.32 & 81.49 \\
\midrule
\multirow{6}{*}{\rotatebox[origin=c]{90}{\emph{netscience}}} &
\multirow{3}{*}{\rotatebox[origin=c]{90}{Plurality}}
& 0.1
& 0.32 & 0.47 & 0.35 & 0.47 & 0.38 & 0.46 & -7.98 & 22.00 & -6.65 & 21.73 & -5.54 & 21.76 \\
& & 0.5
& 0.64 & 0.42 & 0.82 & 0.32 & 0.88 & 0.31 & 4.14 & 21.87 & 13.25 & 21.39 & 20.42 & 21.12 \\
& & 1.0
& 0.87 & 0.31 & 0.92 & 0.24 & 0.96 & 0.13 & 19.21 & 21.01 & 38.58 & 20.78 & 53.11 & 22.03 \\
\cline{2-15}
& \multirow{3}{*}{\rotatebox[origin=c]{90}{Borda}} 
& 0.1
& 0.40 & 0.46 & 0.44 & 0.46 & 0.47 & 0.46 & -21.14 & 74.37 & -18.69 & 74.14 & -16.79 & 74.47 \\
& & 0.5
& 0.63 & 0.48 & 0.69 & 0.46 & 0.74 & 0.43 & 4.71 & 73.90 & 20.49 & 73.70 & 34.58 & 73.56 \\
& & 1.0
& 0.74 & 0.42 & 0.85 & 0.33 & 0.90 & 0.29 & 36.54 & 74.51 & 73.44 & 73.24 & 101.91 & 73.86 \\
\midrule
\multirow{6}{*}{\rotatebox[origin=c]{90}{\emph{polblogs}}} &
\multirow{3}{*}{\rotatebox[origin=c]{90}{Plurality}}
& 0.1
& 0.50 & 0.44 & 0.56 & 0.44 & 0.60 & 0.42 & 3.46 & 21.11 & 6.45 & 20.71 & 7.58 & 20.87 \\
& & 0.5
& 1.00 & 0.00 & 1.00 & 0.00 & 1.00 & 0.00 & 56.12 & 20.49 & 74.57 & 18.99 & 82.99 & 18.32 \\
& & 1.0
& 1.00 & 0.00 & 1.00 & 0.00 & 1.00 & 0.00 & 125.56 & 19.15 & 159.31 & 16.60 & 180.31 & 19.28 \\
\cline{2-15}
& \multirow{3}{*}{\rotatebox[origin=c]{90}{Borda}} 
& 0.1
& 0.54 & 0.46 & 0.62 & 0.44 & 0.65 & 0.45 & 10.63 & 63.13 & 15.99 & 62.35 & 19.55 & 61.99 \\
& & 0.5
& 0.98 & 0.04 & 1.00 & 0.01 & 1.00 & 0.00 & 126.71 & 58.98 & 163.79 & 62.94 & 180.19 & 59.44 \\
& & 1.0
& 1.00 & 0.00 & 1.00 & 0.00 & 1.00 & 0.00 & 271.19 & 62.85 & 339.74 & 59.25 & 375.73 & 55.27 \\
\bottomrule
\end{tabular}

\footnotesize{$\mu$ and $\sigma$ are, respectively, the mean and the standard deviation of the observations averaged over the $10$ preference list permutations.}
\end{table*}
\end{landscape}
\begin{landscape}
\begin{table*}[pt!]
\centering
\caption{\POV and \MOV values relative to the experiments with $m=10$.}
\label{tab:experiments_10}
\begin{tabular}{clrrrrrrrrrrrrr}
\toprule
& & & \multicolumn{6}{c}{\POV} & \multicolumn{6}{c}{\MOV} \\
\cmidrule(lr){4-9}
\cmidrule(lr){10-15}
& & & \multicolumn{2}{c}{$B=5$} & \multicolumn{2}{c}{$B=10$} & \multicolumn{2}{c}{$B=15$} 
& \multicolumn{2}{c}{$B=5$} & \multicolumn{2}{c}{$B=10$} & \multicolumn{2}{c}{$B=15$} \\ 
& & \multicolumn{1}{c}{$\alpha$} & \multicolumn{1}{c}{$\mu$} & \multicolumn{1}{c}{$\sigma$} & \multicolumn{1}{c}{$\mu$} & \multicolumn{1}{c}{$\sigma$}  & \multicolumn{1}{c}{$\mu$} & \multicolumn{1}{c}{$\sigma$} & \multicolumn{1}{c}{$\mu$} & \multicolumn{1}{c}{$\sigma$} & \multicolumn{1}{c}{$\mu$} & \multicolumn{1}{c}{$\sigma$} & \multicolumn{1}{c}{$\mu$} & \multicolumn{1}{c}{$\sigma$} \\ 
\midrule
\multirow{6}{*}{\rotatebox[origin=c]{90}{\emph{facebook}}} &
\multirow{3}{*}{\rotatebox[origin=c]{90}{Plurality}}
& 0.1
& 0.96 & 0.12 & 0.97 & 0.08 & 0.97 & 0.08 & 27.99 & 14.88 & 31.81 & 15.55 & 32.16 & 15.50 \\
& & 0.5
& 1.00 & 0.00 & 1.00 & 0.00 & 1.00 & 0.00 & 183.94 & 12.90 & 217.54 & 13.77 & 220.02 & 15.09 \\
& & 1.0
& 1.00 & 0.00 & 1.00 & 0.00 & 1.00 & 0.00 & 378.91 & 13.12 & 450.20 & 15.08 & 454.48 & 14.53 \\
\cline{2-15}
& \multirow{3}{*}{\rotatebox[origin=c]{90}{Borda}} 
& 0.1
& 0.80 & 0.39 & 0.84 & 0.33 & 0.84 & 0.33 & 111.75 & 133.60 & 130.78 & 134.36 & 134.20 & 135.25 \\
& & 0.5
& 1.00 & 0.00 & 1.00 & 0.00 & 1.00 & 0.00 & 816.22 & 134.28 & 973.03 & 128.19 & 987.42 & 134.47 \\
& & 1.0
& 1.00 & 0.00 & 1.00 & 0.00 & 1.00 & 0.00 & 1701.58 & 133.16 & 2036.01 & 131.07 & 2055.55 & 131.70 \\
\midrule
\multirow{6}{*}{\rotatebox[origin=c]{90}{\emph{irvine}}} &
\multirow{3}{*}{\rotatebox[origin=c]{90}{Plurality}}
& 0.1
& 0.40 & 0.38 & 0.55 & 0.34 & 0.62 & 0.31 & -1.04 & 11.82 & 4.09 & 12.00 & 6.14 & 12.05 \\
& & 0.5
& 0.98 & 0.04 & 1.00 & 0.00 & 1.00 & 0.00 & 51.49 & 12.23 & 79.38 & 11.78 & 91.22 & 15.15 \\
& & 1.0
& 1.00 & 0.00 & 1.00 & 0.00 & 1.00 & 0.00 & 118.46 & 9.50 & 175.53 & 18.33 & 203.05 & 19.76 \\
\cline{2-15}
& \multirow{3}{*}{\rotatebox[origin=c]{90}{Borda}} 
& 0.1
& 0.36 & 0.45 & 0.42 & 0.43 & 0.46 & 0.42 & -25.12 & 101.84 & -4.95 & 99.47 & 3.34 & 99.69 \\
& & 0.5
& 0.93 & 0.13 & 1.00 & 0.01 & 1.00 & 0.00 & 216.60 & 100.45 & 347.52 & 98.22 & 408.77 & 107.29 \\
& & 1.0
& 0.99 & 0.02 & 1.00 & 0.00 & 1.00 & 0.00 & 508.13 & 114.34 & 780.43 & 102.47 & 885.19 & 125.87 \\
\midrule
\multirow{6}{*}{\rotatebox[origin=c]{90}{\emph{netscience}}} &
\multirow{3}{*}{\rotatebox[origin=c]{90}{Plurality}}
& 0.1
& 0.18 & 0.36 & 0.23 & 0.39 & 0.24 & 0.41 & -13.66 & 12.85 & -12.97 & 12.84 & -12.41 & 12.90 \\
& & 0.5
& 0.31 & 0.48 & 0.39 & 0.45 & 0.53 & 0.41 & -6.05 & 13.20 & -0.57 & 13.21 & 3.75 & 13.24 \\
& & 1.0
& 0.51 & 0.41 & 0.86 & 0.27 & 0.97 & 0.09 & 3.17 & 13.69 & 15.00 & 13.33 & 25.42 & 13.36 \\
\cline{2-15}
& \multirow{3}{*}{\rotatebox[origin=c]{90}{Borda}} 
& 0.1
& 0.11 & 0.31 & 0.12 & 0.32 & 0.14 & 0.32 & -108.47 & 81.50 & -105.35 & 81.60 & -102.35 & 81.83 \\
& & 0.5
& 0.20 & 0.41 & 0.25 & 0.41 & 0.32 & 0.42 & -75.12 & 81.05 & -53.48 & 80.35 & -36.07 & 81.25 \\
& & 1.0
& 0.33 & 0.41 & 0.54 & 0.46 & 0.66 & 0.42 & -32.80 & 79.95 & 14.51 & 83.71 & 48.93 & 84.95 \\
\midrule
\multirow{6}{*}{\rotatebox[origin=c]{90}{\emph{polblogs}}} &
\multirow{3}{*}{\rotatebox[origin=c]{90}{Plurality}}
& 0.1
& 0.41 & 0.46 & 0.45 & 0.47 & 0.48 & 0.48 & -2.22 & 15.10 & -0.83 & 15.12 & 0.14 & 14.97 \\
& & 0.5
& 0.97 & 0.06 & 1.00 & 0.01 & 1.00 & 0.00 & 31.21 & 14.73 & 41.51 & 14.67 & 46.43 & 13.25 \\
& & 1.0
& 1.00 & 0.00 & 1.00 & 0.00 & 1.00 & 0.00 & 73.13 & 14.60 & 94.24 & 11.49 & 106.88 & 12.39 \\
\cline{2-15}
& \multirow{3}{*}{\rotatebox[origin=c]{90}{Borda}} 
& 0.1
& 0.25 & 0.41 & 0.27 & 0.44 & 0.28 & 0.45 & -48.35 & 92.53 & -42.21 & 92.46 & -38.94 & 92.17 \\
& & 0.5
& 0.87 & 0.18 & 0.96 & 0.07 & 0.99 & 0.02 & 103.80 & 94.58 & 141.53 & 93.96 & 170.65 & 92.91 \\
& & 1.0
& 1.00 & 0.00 & 1.00 & 0.00 & 1.00 & 0.00 & 293.93 & 95.50 & 372.21 & 100.52 & 423.54 & 93.06 \\
\bottomrule
\end{tabular}

\footnotesize{$\mu$ and $\sigma$ are, respectively, the mean and the standard deviation of the observations averaged over the $10$ preference list permutations.}
\end{table*}
\end{landscape}

\section{Discussion of Results}
\label{sec:conclusion}

The results in our paper are very significant. 
Nowadays social media are are significant sources of information for voters and the massive usage of these channels for political campaigning is a turning point. 
Potential attackers can manipulate the outcome of elections through the spread of targeted ads and/or fake news. 
Being able to control the information spread can have a great impact, but it is not easy to achieve given that traditional media sources are relatively transparent.
Therefore, it is essential to protect the integrity of electoral processes to ensure the proper operation of democratic institutions.
Our results indicate that social influence is a salient threat to election integrity: We provide an approximation algorithm to maximize the \MOV of a target candidate, that can be used by an attacker to control the election results and is of fundamental importance to protect their fairness.

There is only another paper that focus on the problem of election control through social influence~\cite{wilder2018controlling}.
Compared to it, we consider a more realistic model (LTM instead of ICM) that takes into account the amount of influence that voters exercise on each other.
We believe that our algorithm could be used in real-life scenarios to predict election results and to understand what degree of control has been exercised. 
Our results assume the knowledge of election data that are not available (degrees of influences and preferences of voters), but that can be estimated. Even if such estimation is not easy, experimental results show that greedy has good performances even on real-world datasets where this data are uncertain.
With this respect, we are aware of recent studies that analyze the robustness of greedy w.r.t.\ inaccurate estimations of the degrees of influence; again, ground truth for such quantities is not available and good estimates are hard to get. 
Nevertheless, experimental results on greedy algorithm for Influence Maximization showed that the worst case hardness theoretical results do not necessarily translate into bad performance on real-world datasets~\cite{he2018robust}.


\section{Conclusion and Future Work}
Online social networks are increasingly utilized for political campaigning since specific users can be targeted by advertisement and/or fake news.
We focused on the problem of controlling election through social influence: Given a social network of people willing to vote, we aim at selecting a fixed-size subset of voters such that their influence on the others will change the outcome of the elections, making some specific candidate win or lose.

We described a powerful extension of the Linear Threshold Model, which describes the change of opinions taking into account the amount of exercised influence.
We provided a constant factor approximation algorithm to the problems of \emph{constructive} and \emph{destructive} election control, considering arbitrary scoring rule voting systems, including \emph{plurality} and \emph{borda count}.
\textsc{Greedy} (Algorithm~\ref{alg:greedy}) achieves a $\frac{1}{3}(1 - 1/e)$ approximation ratio in the constructive scenario, since we showed that any scoring function is monotone submodular w.r.t.\ the initial set of active nodes.
Similarly, we get a $\frac{1}{2}(1 - 1/e)$ approximation ratio for the destructive scenario.

We performed a simulation of our algorithm in our model, examining it on real-world networks using synthetic election data, i.e., random degrees of influences of voters on each other and random preference lists for each voter. 
We ran the simulation with different combinations of parameters, varying $B$, $\alpha$, $|C|$, and $\pi_v$ for each $v \in V$ on 4 networks exhibiting heterogeneous topologies, namely \textit{facebook}, \textit{polblogs}, \textit{irvine}, \textit{netscience}. 
We observed that \textsc{Greedy} is able to find a solution that makes the target candidate win the election in the plurality rule with 10 candidates between 50\% and 88\% of the times, depending on the value of $\alpha$ that changes the degree of influence, using only 5 seed nodes. 

As future research directions we would like to further study our model in a wider range of scenarios which are not currently captured, including \emph{multi-winner} and \emph{proportional representation} systems.
We also believe that approaches that mix constructive and destructive control could be analyzed to get better approximation ratios.
Moreover, we would like to extend our model in order to consider a more uncertain scenario, in which the preferences of voters are not known.
Finally, it would be interesting to study how to prevent election control for the integrity of voting processes, e.g., through the placement of monitors in the network~\cite{zhang2015monitor,amoruso2017monitor} or by considering strategic settings~\cite{yin2018optimal,wilder2018defending}.

\clearpage
\bibliographystyle{alpha}
\bibliography{references}

\end{document}